\documentclass{article}
\usepackage[utf8]{inputenc}
\usepackage[english]{babel}
\usepackage[a4paper, total={6in, 8in}]{geometry}
\usepackage{xcolor}
\usepackage{xspace}
\usepackage{soul}
\usepackage{indentfirst}
\usepackage{url}
\usepackage{amsmath}
\usepackage{amssymb}
\usepackage{amsthm}
\usepackage{bm}
\usepackage{bbm}
\usepackage{subcaption}
\usepackage{cleveref}
\usepackage{listings}
\usepackage{verbatim}
\usepackage{color} 
\definecolor{mygreen}{RGB}{28,172,0} 
\definecolor{mylilas}{RGB}{170,55,241}
\usepackage{tikz}
\usetikzlibrary{trees}

\usepackage{algorithm}
\usepackage{algpseudocode}

\newtheorem{assumption}{Assumption}

\newtheorem{lemma}{Lemma}

\newtheorem{proposition}{Proposition}

\DeclareMathOperator*{\argmin}{arg\,min}
\DeclareMathOperator*{\diag}{diag}

\newcommand{\Cov}{\mathrm{Cov}}

\newcommand{\btheta}{\boldsymbol{\theta}}
\newcommand{\bTheta}{\boldsymbol{\Theta}}

\def\beq{ \begin{equation} }
\def\eeq{ \end{equation} }

\def\square{\vcenter{\vbox{\hrule height .4pt
  \hbox{\vrule width .4pt height 5pt \kern 5pt
        \vrule width .4pt} \hrule height .4pt}}}

\def\sqz{\kern-0.2em}

\newcommand{\E}{\mathbb{E}}
\newcommand{\indi}{\mathbbm{1}}

\newcommand{\tb}{\color{blue}}

\usepackage{mathtools}

\NewDocumentCommand{\expect}{ e{^} s o >{\SplitArgument{1}{|}}m }{%
  \operatorname{E}
  \IfValueT{#1}{{\!}^{#1}}
  \IfBooleanTF{#2}{
    \expectarg*{\expectvar#4}%
  }{
    \IfNoValueTF{#3}{
      \expectarg{\expectvar#4}%
    }{
      \expectarg[#3]{\expectvar#4}%
    }%
  }%
}
\NewDocumentCommand{\expectvar}{mm}{%
  #1\IfValueT{#2}{\nonscript\;\delimsize\vert\nonscript\;#2}%
}
\DeclarePairedDelimiterX{\expectarg}[1]{[}{]}{#1}

\title{A statistical framework for detecting therapy-induced resistance from drug screens}
\author{Chenyu Wu, Einar Bjarki Gunnarsson, Jasmine Foo, Kevin Leder}
\date{\today}

\graphicspath{{Plots/}}

\begin{document}
\lstset{language=Matlab,%
    breaklines=true,%
    morekeywords={matlab2tikz},
    keywordstyle=\color{blue},%
    morekeywords=[2]{1}, keywordstyle=[2]{\color{black}},
    identifierstyle=\color{black},%
    stringstyle=\color{mylilas},
    commentstyle=\color{mygreen},%
    showstringspaces=false,
    numbers=left,%
    numberstyle={\tiny \color{black}},
    numbersep=9pt, 
    emph=[1]{for,end,break},emphstyle=[1]\color{red}, 
}

\maketitle

\begin{abstract}
    Resistance to therapy remains a significant challenge in cancer treatment, often due to the presence of a stem-like cell population that drives tumor recurrence post-treatment. 
    Moreover, many anticancer therapies induce plasticity, converting initially drug-sensitive cells to a more resistant state, e.g. through epigenetic processes and de-differentiation programs.  Understanding the balance between therapeutic anti-tumor effects and induced resistance is critical for identifying treatment strategies.     In this study, we introduce a robust statistical framework, based on multi-type branching process models of the evolutionary dynamics of tumor cell populations, to detect and quantify therapy-induced resistance phenomena from high throughput drug screening data. 
    Through comprehensive \textit{in silico} experiments, we show the efficacy of our framework in estimating parameters governing population dynamics and drug responses in a heterogeneous tumor population where cell state transitions are influenced by the drug. 
    Finally, using recent \textit{in vitro} data from multiple sources, we demonstrate that our framework is effective for analyzing real-world data and generating meaningful predictions.
\end{abstract}

\section{Introduction}

Cellular plasticity, the ability to reversibly transition between phenotypic states, is a fundamental cell feature. 
In the context of cancer treatment, cellular plasticity often manifests as a transition from a therapy-sensitive state to a therapy-resistant state, a phenomenon known as therapy-induced resistance \mbox{\cite{shi2023tumor, boumahdi2020great}}.
This phenomenon has been observed not only in pharmaceutical treatments but also in other modalities such as radiotherapy \mbox{\cite{kyjacova2015radiotherapy,iborra2023chemotherapy,shi2023tumor}}, reducing the efficacy of a broad class of anti-cancer therapies. 
A critical phenotypic state transition related to therapy-induced resistance is the de-differentiation of cancer non-stem-like cells (CNSCs).
This process involves mature cancer cell phenotypes (CNSCs) reverting to immature, cancer stem-like phenotypes (CSCs), which have unlimited proliferation potential and are inherently drug-resistant  \mbox{\cite{dingli2006successful}}.
Understanding drug-induced transitions from CNSCs to CSCs can help shed light on tumor heterogeneity and lead to improved treatment plans \mbox{\cite{leder2014mathematical}}.

In a recent study \cite{padua2023high}, ciclopirox olamine (CPX-O), originally an antifungal agent, was identified for its ability to induce the production of gastric CSCs from gastric CNSCs. 
The study utilized the SORE-GFP reporter system coupled with fluorescence-activated cell sorting (FACS) to distinguish between gastric CSCs and CNSCs. 
A high-throughput screening (HTS) technique was employed to assess the effects of various potential medications on a homogeneous population of gastric CNSC to explore drug-induced plasticity. 
One key assumption of this study is the effective separation of gastric CSCs and CNSCs using the FACS technique and the SORE-GFP reporter system, which may limit the scope of the findings. 
Additionally, the study noted CPX's significant cytotoxic impact on gastric CNSCs.
This complicates assessment of the level of drug-induced plasticity, since increased abundance of CSCs relative to CNSCs can both be explained by drug-induced plasticity and a selective advantage over CNSCs.
Consequently, the empirical inference of drug-induced plasticity presents significant challenges. As an alternative approach, we propose using a mathematical model to understand how the drug affects the proliferation of both CNSCs and CSCs, and to tease apart the drug's effect on cell proliferation vs.~transitions from CNSCs to CSCs.


Mathematical modeling has increasingly played a crucial role in understanding and treating cancer \cite{yin2019review,mathur2022optimizing}, offering advantages in modeling complex biological dynamics through manageable mathematical models. 
For instance, recent studies have applied the multi-type branching process model to depict diverse tumor populations undergoing phenotypic switching \cite{jilkine2019mathematical,gunnarsson2023statistical,gunnarsson2020understanding}. 
This approach abstracts the cell division process into `branching events', allowing each cell to generate descendants across any subtype described in the model. 
As a result, the model can naturally describe processes such as self-proliferation, differentiation, and de-differentiation of CSCs and CNSCs.

On the other hand, several recent studies \cite{kohn2023phenotypic,wu2024using} have proposed frameworks for deconvoluting subpopulation structure directly from HTS bulk cell count data, avoiding the need for a reporter system coupled with FACS to separate distinct subpopulations.
Instead, these studies distinguish subpopulations based on their distinct responses to a given drug. 
For instance, these works have successfully identified Imatinib-sensitive and -resistant Ba/F3 cells from a population mixture by examining heterogeneous drug response curves at various concentration levels.
However, these frameworks assume that each subpopulation can only generate its own subtype, thereby limiting their ability to capture the differentiation and de-differentiation dynamics of CSCs and CNSCs.

In this study, we have developed a novel statistical framework that integrates the multi-type branching process with a drug response model to simulate heterogeneous cell growth dynamics under drug influence. 
This innovative framework serves as a mathematical tool for investigating drug-induced plasticity between CNSCs and CSCs directly from HTS bulk cell count data.
Given our focus on drug-induced plasticity between CNSCs and CSCs and the inherent resistance of CSCs, we will use the term drug-induced plasticity and the term drug-induced resistance interchangeably.

The paper is organized as follows.
In Section \ref{sec:Methodology}, we introduce our newly proposed framework. 
Based on this framework, we conduct multiple \textit{in silico} experiments to validate our model predictions in Section \ref{sec:In silico}. 
Subsequently, in Section \ref{sec:In vitro}, we validate our framework using \textit{in vitro} data obtained from \cite{padua2023high} and \cite{comandante2020phenotype}. 
Finally, in Section \ref{sec:Conclusion}, we conclude by discussing the advantages and limitations of our framework.
A detailed derivation of the methods and the experimental settings are provided in the Appendix.

\section{Methodology}

\label{sec:Methodology}


In this section, we start by introducing a general model to describe heterogeneous tumor growth dynamics, capable of accommodating an arbitrary number of distinct phenotypes.
Then, we develop a drug-effect model using the classical Hill equation, which is widely employed to represent dose-response relationships. 
By integrating these two models, we propose a novel statistical framework for inferring underlying tumor dynamics under the treatment, including homogeneous proliferation and phenotypic transition, from HTS data. 
Finally, we tailor this framework to a specific case involving two subpopulations within the tumor, CSCs and CNSCs, which will be the focus of our analysis in the remainder of the paper.
In what follows, we denote vectors and matrices in boldface letters and all the vectors are row vectors.

\subsection{Asymmetrical birth multi-type branching process model}

\label{sec:Cell dynamic model}


Assume there are $K$ subpopulations of phenotypes within a tumor, each exhibiting its own growth dynamics and drug response. 
These subpopulations are indexed by $\mathcal{K} = \{1,\cdots, K\}$. 
Increasing evidence indicates that asymmetric cell division is a key mechanism for both differentiation \cite{hitomi2021asymmetric,morrison2006asymmetric} and de-differentiation \cite{song2021asymmetric}. 
Therefore, we assume that each phenotype can transition to one or more other phenotypes via asymmetric cell division.
Our framework can also readily accommodate phenotypic transitions that occur independently of cell divisions, as discussed for example in \cite{gunnarsson2023statistical}.

We employ an asymmetric birth multi-type branching process model in continuous time \cite{athreya2004branching}.
In the model, a type-$i$ cell divides into two type-$i$ cells at rate $\alpha_i$, experiences natural death at rate $\beta_i$, and asymmetrically divides into one type-$i$ cell and one type-$j$ cell at rate $\nu_{ij}$ (see Figure \ref{fig:Multi-type branching illustrative}). 
Specifically, within a given infinitesimal time interval $\Delta t > 0$, a type-$i$ cell has a probability $\alpha_i \Delta t$ of dividing into two type-$i$ cells, $\beta_i \Delta t$ of dying, and $\nu_{ij}\Delta t$ of asymmetrically dividing into one type-$i$ and one type-$j$ cell. 
The net growth rate of type-$i$ cells, $\kappa_i$, is defined as $\kappa_i := \alpha_i - \beta_i$.
The dynamics of the model are encoded in the  \textit{infinitesimal generator matrix} {\bf A}, which is the $K \times K$ matrix
\begin{equation}
\label{eq: generator matrix}
\mathbf{A} := \begin{bmatrix}
\kappa_1 & \nu_{12} & \cdots & \nu_{1K}\\
\nu_{21} & \kappa_2 & \cdots & \nu_{2K}\\
\vdots & \vdots & \ddots & \vdots\\
\nu_{K1} & \nu_{K2} & \cdots & \kappa_K
\end{bmatrix}.
\end{equation}
In this matrix, the $(i,j)$-th element is the net rate at which a type-$i$ cell produces a type-$j$ cell.
For all $K$ phenotypes, we assume a strictly positive birth rate ($\alpha_i > 0, i \in \mathcal{K}$), a non-negative death rate ($\beta_i \geq 0, i\in \mathcal{K}$), and a non-negative asymmetric birth rate ($\nu_{i,j} \geq 0, i\neq j$).
We furthermore assume that $\kappa_i>0$ for all $i\in \mathcal{K}$, i.e.~the net birth rate is positive, though transition and death may not occur for all phenotypes. 

For simplicity, we assume that the multi-type branching process is \textit{irreducible}, meaning that each subpopulation can eventually produce descendants in any other subpopulation, possibly through intermediate types. Mathematically, this means that the \textit{infinitesimal generator matrix} $\mathbf{A}$ cannot be transformed into a block upper triangular matrix via simultaneous row or column permutation.



\begin{figure}[ht]
\begin{subfigure}{.33\textwidth}
  \centering
  \includegraphics{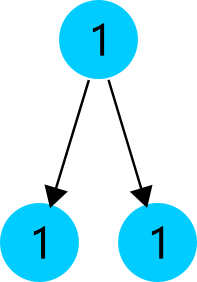}  
  \caption{Symmetric division (rate: $\alpha_1$)}
  \label{fig:Sy_birth}
\end{subfigure}
\begin{subfigure}{.33\textwidth}
  \centering
  \includegraphics{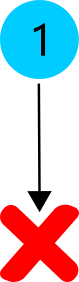}  
  \caption{Death (rate: $\beta_1$)}
  \label{fig:Death}
\end{subfigure}
\begin{subfigure}{.33\textwidth}
  \centering
  \includegraphics{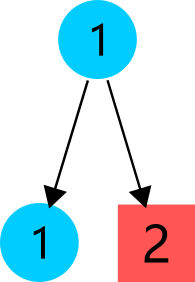}
  \caption{Asymmetric division (rate: $\nu_{12}$)}
  \label{fig:Asy_birth}
\end{subfigure}
\caption{Illustration of three independent events in the asymmetrical birth multi-type branching process model. Different shapes represent different phenotypes.}
\label{fig:Multi-type branching illustrative}
\end{figure}

We denote the cell counts of each phenotype as a random vector: $\mathbf{B}(t) = [B_{1}(t),B_{2}(t),\cdots, B_{K}(t)]$, where $B_{i}(t)$ is the number of type-$i$ cells at the time $t$. 
For the special case where the process is started by a single type-$i$ cell, the cell count vector at time $t$ is denoted by $\mathbf{B}^{(i)}(t) = [B^{(i)}_{1}(t),B^{(i)}_{2}(t),\cdots, B^{(i)}_{K}(t)]$, where $\mathbf{B}^{(i)}(0) = \mathbf{e}_i$ is the $i$-th identity vector.
The \textit{mean vector} and \textit{covariance matrix} for $\mathbf{B}^{(i)}(t)$ are denoted by
\begin{align*}
    \mathbf{m}^{(i)}(t) &:= \expect*{\mathbf{B}^{(i)}(t)},\\
    \mathbf{\Xi}^{(i)}(t)&:= \expect*{\left(\mathbf{B}^{(i)}(t) - \mathbf{m}^{(i)}(t)\right)^\intercal \left(\mathbf{B}^{(i)}(t) - \mathbf{m}^{(i)}(t)\right)},
\end{align*}
where $t\geq 0$ and $\left(\mathbf{B}^{(i)}(t) - \mathbf{m}^{(i)}(t)\right)^{\intercal}$ represents the transpose of the vector $\left(\mathbf{B}^{(i)}(t) - \mathbf{m}^{(i)}(t)\right)$. 
To explicitly compute the \textit{mean vector}, we define the {\em mean matrix} using the matrix exponential:
\[\mathbf{M}(t) = \exp{(t \mathbf{A})} = \sum_{j = 0}^{\infty} \left(\frac{t^j}{j!} \right) \mathbf{A}^{j}.\]
The \textit{mean vector} $\mathbf{m}^{(i)}(t)$ can be computed as the $i$-th row of the {\em mean matrix},
\[\mathbf{m}^{(i)}(t) = \mathbf{e}_i \mathbf{M}(t),\]
and the covariance matrix $\mathbf{\Xi}^{(i)}(t)$ can be computed as shown in Proposition \ref{prop: finite CLT} below.

\subsection{Drug-effect Model}

\label{sec:drug-effect model}
Before modeling the drug effect, we introduce a fixed drug dose $d$ and denote the cell count at time $t$ and concentration level $d$ as $\mathbf{B}(t,d)$. The \textit{mean matrix} $\mathbf{M}(t,d)$ and the covariance matrix $\mathbf{\Xi}^{(i)}(t,d)$ follow. Note that we use concentration level and dose interchangeably.

\subsubsection{Base model}
To model the drug effect, we define the well-known Hill equation with parameters $(b,E,m)$ as
\[H(d;b,E,m) = b + \frac{1-b}{1 + (d/E)^m}.\]
\begin{sloppypar}
The Hill equation is a classic sigmoidal function used to describe a dose-response \cite{prinz2010hill}. 
In the equation, the parameter $b$ represents the maximum drug effect, since 
$H(0;b,E,m) = 1$ and
$\lim_{d \rightarrow \infty} H(d;b,E,m) = b.$
For $b\in (0,1)$, the Hill equation is strictly decreasing, while it is strictly increasing for $b>1$ (Figure \ref{fig:Hill equations}).
The parameter $E$ indicates the concentration at which 50 percent of the maximum effect is attained, known as the inflection point. 
The last parameter $m$ is the Hill coefficient, which controls the steepness of the Hill equation around the inflection point. 
As a simplification, we fix the Hill parameter $m = 1$ for all drug effects when conducting our {\em in silico} experiments in Section \ref{sec:In silico}. 
We then reintroduce this parameter when modeling \textit{in vitro} data in Section \ref{sec:In vitro}. 
\end{sloppypar}

\begin{figure}[ht]
\begin{subfigure}{.5\textwidth}
  \centering
  \includegraphics[width = \linewidth]{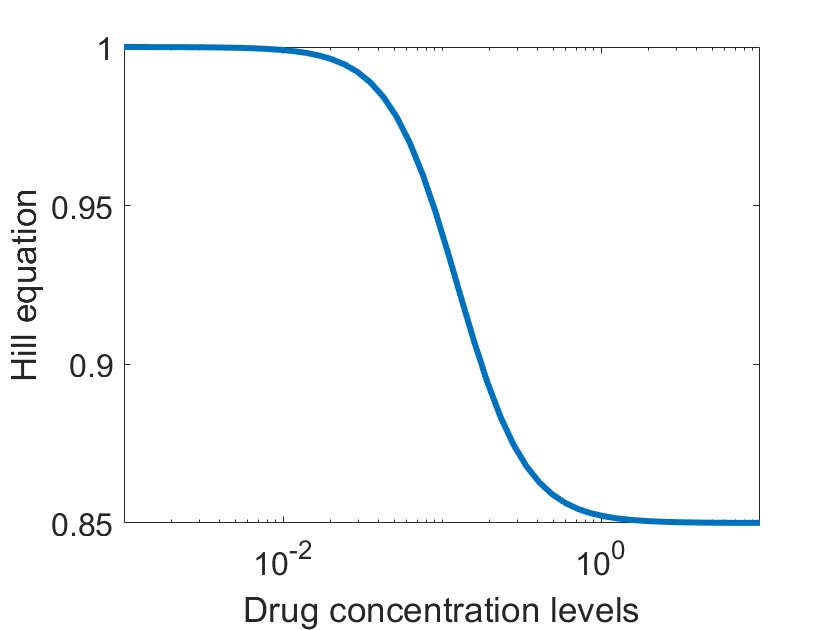}  
  \caption{Hill equation when $b \in (0,1)$}
  \label{fig:Hill(0,1)}
\end{subfigure}
\begin{subfigure}{.5\textwidth}
  \centering
  \includegraphics[width = \linewidth]{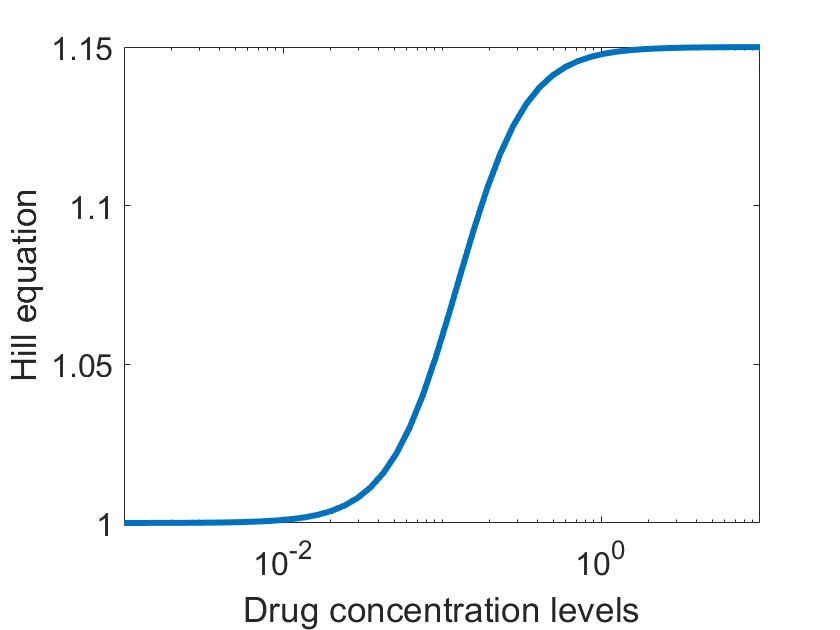}  
  \caption{Hill equation when $b > 1$}
  \label{fig:Hill(1,2)}
\end{subfigure}
\caption{Hill equation illustrations}
\label{fig:Hill equations}
\end{figure}



We initially assume the effects of a single drug are time-homogeneous.
We distinguish between two types of drug responses, drug toxicity response and drug-induced plasticity response, thereby capturing multiple effects from a single drug.



\begin{itemize}

\item \textit{Drug toxicity response (cytotoxic effect)}: The parameters related to this response are denoted as $(b_{i,\beta},E_{i,\beta})$ for type $i$. 
We assume that $b_{i,\beta} \in (0,1)$ and that the death rate of type-$i$ under dose $d$ is given by (for the case $m=1$):
\begin{equation}
\label{eq: Drug affected symmetric birth rate}
\beta_i(d) = \beta_i - \log(H(d;b_{i,\beta},E_{i,\beta})) = \beta_i - \log\left(b_{i,\beta} + \frac{1 - b_{i,\beta}}{1 + (d/E_{i,\beta})} \right).
\end{equation}
The net growth rate of a type $i$ cell, $\kappa_i$, is therefore negatively affected by increasing drug concentration, i.e.~$\kappa_i(d) = (\alpha_i - \beta_i) + \log(H(d;b_{i,\beta},E_{i,\beta}))$.
In this scenario, the drug is assumed to act through a cytotoxic mechanism, meaning that higher doses result in increased rates of cell death. 
We note that our framework can also easily incorporate cytostatic effects, where higher doses lead to a lower cell division rate, using a similar Hill function.


\item \textit{Drug-induced plasticity response}: The parameters related to this response are denoted as $(b_{i,\nu},E_{i,\nu})$ for type-$i$. 
For simplicity, we assume here that the drug effect on type-$i$ cell transitions is independent of the target phenotype $j$.
Drug-induced phenotypic transition is a widely observed effect of cancer treatment \cite{pisco2013non,goldman2015temporally,su2017single}. 
Particularly, drug-sensitive phenotypes may be induced to develop resistance when exposed to the drug. 
When modeling the drug effect on the asymmetric birth rate, we assume $b_{i,\nu} \geq 1$ to account for elevated transitions between subpopulations caused by the therapeutic environment. 
We furthermore assume that the asymmetric birth rate of type $i$ to any other type $j$  under dose $d$ is given by (for the case $m=1$):
\begin{equation}
\label{eq: Drug affected asymmetric birth rate}
\nu_{ij}(d)= \nu_{ij} + \log(H(d;b_{i,\nu},E_{i,\nu})) = \nu_{ij} + \log\left(b_{i,\nu} + \frac{1 - b_{i,\nu}}{1 + (d/E_{i,\nu})}\right).
\end{equation}
The assumption $b_{i,\nu} \geq 1$ incorporates the situation where the drug does not impact phenotypic transitions, when $b_{i,\nu} = 1$, $\nu_{ij}(d) = \nu_{ij}$ for all $d \geq 0$. 
It is worth mentioning that our framework can also model drug inhibition of the asymmetric birth rate by letting $b_{i,\nu} \in (0,1)$.

\end{itemize}

\subsubsection{Logistic time-delayed drug response}
\label{sec:time-delayed model}

Our framework can be extended to accommodate a more complex drug-effect model, which may involve time-inhomogeneous drug effects.
In particular, for the {\em in vitro} dataset from \cite{padua2023high} investigated in Section \ref{sec:In vitro}, there appear to be delayed drug effects, which have been observed in other studies as well \cite{harris2016unbiased, russo2022modified, yang2022mathematical}. 
To address this behavior, we allow for the possibility of a time-inhomogeneous drug effect, which is modeled using a two-parameter logistic function.
In this case, we assume that the maximum drug effect parameter $b$ depends on time $t$ through
$$ |1-b(t)| = \frac{1}{1 + \exp(k(t_0 - t))}|1-b|,$$
where $k,t_0$ are logistic function parameters.
In this way, we can model the drug's maximum potential $b(t)$ gradually reaching its ultimate value as time increases.
Absolute values are taken to allow both the cytotoxic effect $b(t) < 1$ and the drug-induced plasticity effect $b(t) > 1$ to be time-inhomogeneous.

Notably, incorporating this time-dependent drug effect model introduces a time-dependent \textit{infinitesimal generator matrix} $\mathbf{A}(t,d)$, complicating the computation of the \textit{mean matrix} and the \textit{covariance matrix}. 
Detailed implementation can be found in Secion \ref{sec: Deterministic approximation and time-delayed drug-effect model}.



\subsection{Long-run behavior}


\label{sec: Long-run behavior}

As the HTS experiments being considered are run for a longer time duration,
it becomes important to understand the long-run behavior of the process $\mathbf{B}(t,d)$. 
According to the Perron-Frobenius Theorem, 
the irreducible \textit{infinitesimal generator matrix} $\mathbf{A}(d)$ has a strictly positive largest eigenvalue $\lambda_1(d)$,
which corresponds to the largest eigenvalue $e^{\lambda_1(d) t}$ of the mean matrix $\mathbf{M}(t,d)$. 
Additionally, the corresponding left eigenvector $\bm{\pi}(d)$ is strictly positive. 
A well-known limiting result reviewed in \cite{athreya2004branching} states that there exists, almost surely, a non-negative numerical random variable $W$ such that 
\[\lim_{t\rightarrow \infty} \mathbf{B}(t,d) e^{-\lambda_1(d) t} = W\bm{\pi}(d).\]
This result characterizes two aspects of the long-run behavior of the process $\mathbf{B}(t,d)$.
First, the stochastic process $\mathbf{B}(t,d)$ will grow with a deterministic exponential \textit{long-run growth rate} $\lambda_1(d)$.
Second, when the eigenvector $\bm{\pi}(d)$ is normalized so that $\sum_{i=1}^K \pi_i(d) = 1$, then $\pi_i(d)$ describes the long-run proportion of type-$i$ cells in the population. 
Therefore, we refer to the normalized version of $\bm{\pi}(d)$ as the {\em stable proportion} between phenotypes.

\subsection{Statistical Model}

\label{sec:Statistical model}



Now, we construct a statistical framework to infer heterogeneous tumor growth dynamics and drug responses using HTS live-cell imaging bulk data. 
Here, `bulk data' refers to aggregated total cell counts across all subpopulations within the heterogeneous tumor.
Our framework is designed to accommodate data collected via live-cell imaging techniques, allowing efficient capture of bulk cell counts from a single sample across multiple time points \cite{nketia2017analysis}. 

Assume the bulk dataset $\mathbf{X}$ is collected across a set of drug concentration levels $\mathcal{D} = \{d_1,\cdots, d_{N_D}\}$ and a set of time points $\mathcal{T} = \{t_1,\cdots,t_{N_T}\}$, where $0< t_1 <\cdots<t_{N_T}$.
For each drug dose $d \in \mathcal{D}$, $N_R$ samples are cultivated, and live-cell imaging technique is used to collect bulk cell count data at time points $t_1,\ldots,t_{N_T}$.
Let $\mathbf{x}_{\mathcal{T},d,r} = (x_{1,d,r},\ldots,x_{N_t,d,r})$ denote the data collected for the $r$-th replicate under drug concentration $d$, 
where $x_{k,d,r}$ is the bulk cell count at time point $t_k$.
Taking into account experimental measurement error, we propose the statistical model
\begin{equation}
    \label{eq: Data model}
    \mathbf{x}_{\mathcal{T},d,r} = \mathbf{Y}_{d,r}(\mathcal{T}) +\mathbf{Z}_{d,r},
\end{equation}
where $\mathbf{Y}_{d,r}(\mathcal{T})$ denotes bulk cell counts at the time points ${\cal T}$ under the drug-affected evolution dynamics model outlined in Sections \ref{sec:Cell dynamic model} and \ref{sec:drug-effect model}, and $\mathbf{Z}_{d,r}\sim N(0,c^2 \mathbf{I})$ are independent multivariate normally distributed noise terms.
In other word, $\mathbf{Y}_{d,r}(\mathcal{T})$ is a random vector with $N_T$ elements, each obtained by summing the random vector $\mathbf{B}(t_i,d)$ for $i = 1,\cdots,N_T $.
Assuming that $n_i$ is the starting number of cells of type $i$ in each experiment,
we can write
\begin{equation}
    \label{eq: Total cell time vector}
    \mathbf{Y}_{d,r}(\mathcal{T}) := \sum_{i = 1}^{K} \sum_{j = 1}^{n_i} \mathbf{y}_j^{(i)}(\mathcal{T},d,r),
\end{equation}
where $\mathbf{y}_j^{(i)}(\mathcal{T},d,r)$ denotes the size of the clone started by the $j$-th initial type $i$ cell at the time points ${\cal T}$ under the stochastic model.
Let $n$ denote the initial total cell count and $\mathbf{p} = [p_1,p_2,\cdots,p_K]$ denote the initial subpopulation distribution, with $n_i = np_i$ for all $i$.
In this study, we assume that $p_i$ is independent of $n$ for all $i$, implying that $n_i \rightarrow \infty$ as $n \rightarrow \infty$.

\subsubsection{Central limit theorem}

For simplicity, we temporarily omit the explicit notation for drug concentration levels and experimental replicates and write $\mathbf{Y}(\mathcal{T})$ and $\mathbf{y}_j^{(i)}(\mathcal{T})$. 
When bulk cell counts are observed as opposed to individual subpopulation counts, the stochastic process $\mathbf{Y}(t)$ is no longer a Markov process.
Therefore, the exact probability distribution for $\mathbf{Y}(\mathcal{T})$ becomes overly complex \cite{wu2024using}.
Given that drug screening experiments are usually started by a relatively large number of cells (in \cite{dravid2021optimised} for example, 2500 cells were deployed in each well), it is natural to focus on the asymptotic behavior of $\mathbf{Y}(\mathcal{T})$ as the initial cell number increases ($n\to\infty$). 
For that, we need to understand the asymptotic behavior of $\sum_{j = 1}^{n_i} \mathbf{y}_j^{(i)}(\mathcal{T})$.
We denote the expectation
\[\expect*{\mathbf{y}_j^{(i)}(\mathcal{T})} := \bm{\mu}^{(i)}(\mathcal{T}) = [\langle \mathbf{m}^{(i)}(t_1),\mathbf{1}\rangle, \cdots, \langle \mathbf{m}^{(i)}(t_{N_T}),\mathbf{1}\rangle],\]
where $\mathbf{1}$ is a size $K$ all ones vector, $\langle \cdot,\cdot \rangle$ is the inner product, and ${\bf m}^{(i)}(t)$ denotes the mean vector described in Section \ref{sec:Cell dynamic model}. We furthermore define a centered and normalized process for each $i\in \mathcal{K}$ by
\[\mathbf{W}_{n_i}^{(i)}(\mathcal{T}) := \frac{1}{\sqrt{n_i}} \sum_{j = 1}^{n_i} \left(\mathbf{y}^{(i)}_{j}(\mathcal{T}) - \bm{\mu}^{(i)}(\mathcal{T}) \right).\]
A direct application of the multivariate central limit theorem gives the following result, where `$\Rightarrow$' denotes converge in distribution. 
\begin{proposition}
    \label{prop: finite CLT}
    As $n_i \rightarrow \infty$:
    \[\mathbf{W}_{n_i}^{(i)}(\mathcal{T}) \Rightarrow N(\bm{0},\mathbf{V}^{(i)}(\mathcal{T}))\]
    where
    \[\mathbf{V}^{(i)}_{a,b}(\mathcal{T}) = \mathbf{1} \mathbf{\Xi}^{(i)}(t_{a}) \mathbf{M}(t_{b} - t_{a}) \mathbf{1}^{\intercal} \quad \text{ for } 1\leq a\leq b\leq N_T.\]
    Here, $\mathbf{M}(t)$ is the \textit{mean matrix}, and $\mathbf{\Xi}^{(i)}(t_{a})$ is a covariance matrix (Section \ref{sec:Cell dynamic model}) given by
    \[\begin{split}
    \mathbf{\Xi}^{(i)}(t)
    &= \diag(\mathbf{m}^{(i)}(t)) - \mathbf{m}^{(i)}(t)^{\intercal} \mathbf{m}^{(i)}(t) + \int_{0}^{t} \left(\mathbf{M}(t-\tau) \right)^{\intercal} \mathbf{C}^{(i)}(\tau) \left(\mathbf{M}(t-\tau) \right) d\tau,
    \end{split}\]
    where
    \[\mathbf{C}^{(i)}_{jk}(\tau) = \begin{cases}
    \nu_{kj}\mathbf{M}_{ik}(\tau) + \nu_{jk}\mathbf{M}_{ij}(\tau) & j\neq k\\
    2 \alpha_{j} \mathbf{M}_{ij}(\tau) & j = k.
    \end{cases}\]
\end{proposition}
This result can be derived from a more general result, which is stated as Proposition \ref{prop: Lindeberg CLT} in Appendix \ref{appx: Lindeberg CLT}. According to Proposition \ref{prop: finite CLT}, we can approximate $\sum_{j = 1}^{n_i}\mathbf{y}_j^{(i)}(\mathcal{T})$ as
\[\sum_{j = 1}^{n_i}\mathbf{y}_j^{(i)}(\mathcal{T}) \approx n_i \bm{\mu}^{(i)}(\mathcal{T}) + N(0, n_i \mathbf{V}^{(i)}(\mathcal{T}))\]
for sufficiently large $n_i$. We can approximate $\mathbf{Y}(\mathcal{T})$ accordingly. It is important to note that this approximation holds true for each fixed drug concentration level $d$.







\subsubsection{Maximum likelihood estimation (MLE) framework}

Now we reintroduce the drug concentration level $d\in \mathcal{D}$ and use the approximation
\[\mathbf{Y}_{d,r}(\mathcal{T}) \approx \sum_{i = 1}^{K} n_i \bm{\mu}^{(i)}(\mathcal{T},d) + N^{(r)}(0,n_i \mathbf{V}^{(i)}(\mathcal{T},d)),\]
where $N^{(r)}(0,n_i \mathbf{V}^{(i)}(\mathcal{T},d))$ represents independent and identical copies of a random vector with distribution $N(0,n_i \mathbf{V}^{(i)}(\mathcal{T},d))$.
We then formulate our newly proposed statistical model as
\begin{equation}
    \label{eq: stat model}
    \mathbf{x}_{\mathcal{T},d,r} \approx \sum_{i = 1}^{K} n_i \bm{\mu}^{(i)}(\mathcal{T},d) + N^{(r)}(0,n_i \mathbf{V}^{(i)}(\mathcal{T},d)) + N^{(r)}(0,c^2 \mathbf{I}),
\end{equation}
where the final term captures experimental measurement error as before.
Based on this statistical model, we compute the likelihood function $L(\btheta(\mathcal{K})|\mathbf{X})$, which represents the probability of obtaining the observed dataset $\mathbf{X}$ given the set $\btheta(\mathcal{K})$ of all parameters in the model:
\[\btheta(\mathcal{K}) = \{(\alpha_i,\beta_i,(\nu_{ij})_{j\neq i},b_{i,\beta},E_{i,\beta},b_{i,\nu},E_{i,\nu})_{ i \in \mathcal{K}}, c\}.\]
We then apply the maximum likelihood estimation framework to obtain the parameter set that is most likely to explain the observed dataset.
The estimated parameter set $\hat{\btheta}(\mathcal{K})$ is computed by minimizing the negative log-likelihood, which is equivalent to maximizing the likelihood function:
\begin{equation}
    \label{eq: minimize negative log likelihood}
    \hat{\btheta}(\mathcal{K}) = \argmin_{\btheta(\mathcal{K}) \in \hat{\bTheta}} - \log(L(\btheta(\mathcal{K})|\mathbf{X})),
\end{equation}
where $\hat{\bTheta}$ is the set of feasible parameters. 
In Table \ref{tab:notions}, we provide an overview of notation.

{\tiny
\begin{table}[ht]
    \centering
    \resizebox{\textwidth}{!}{\begin{tabular}{c|c|l|c}
    \hline 
    \hline
        Notation & Dimension &  Description & Definition/Range \\
         \hline

        $\alpha_i$ & $1$ & Subtype $i$ symmetric division rate & $\alpha >0$\\

        $\beta_i(d)$ & $1$ & Subtype $i$ death rate &  $\beta(0) \geq 0$\\

        $\nu_{ij}(d)$ & $1$ & Subtype $i$ asymmetric division rate to subtype $j$ & $\nu(0) \geq 0$\\

        $\lambda_1(d)$ & $1$ & \textit{Long-run growth rate} & $\lambda_1(0) > 0$\\

        $\mathcal{K}$ & $1\times K$ & Subpopulation index set  & $|\mathcal{K}| = K$\\

        $\mathbf{B}(t,d)$ & $1 \times K$ & Vector of subpopulation counts  & $\mathbf{B}(t,d) \geq \bm{0}$\\

        $\mathbf{B}^{(i)}(t,d)$ & $1\times K$ & $\mathbf{B}(t)$ generated from single type $i$ cell & $\mathbf{B}^{(i)}(t,d) \geq \bm{0}$\\

        $\mathbf{m}^{(i)}(t,d)$ & $1\times K$ & Expected value of $\mathbf{B}^{(i)}(t,d)$ & $\mathbf{m}^{(i)}(t,d) \geq \bm{0}$\\

        $\mathbf{\Xi}^{(i)}(t,d)$ & $K\times K$ & Covariance matrix of $\mathbf{B}^{(i)}(t,d)$ & $det(\mathbf{\Xi}^{(i)}(t,d)) >0$\\

        $\mathbf{A}(d)$ & $K\times K$ & \textit{infinitesimal generator matrix} & $\mathbf{A}_{ii}(d) = \alpha_i - \beta_i(d), \mathbf{A}_{ij} = \nu_{ij}(d)$\\

        $\mathbf{M}(t,d)$ & $K\times K$ & Mean matrix & $\lim_{t \rightarrow 0} \mathbf{M}(t,d) = \mathbf{I}$\\

        $\mathbf{I}$ & $K\times K$ & Identity matrix & $\mathbf{I}$\\

        $\mathbf{C}^{(i)}(t,d)$ & $K\times K$ & Factor matrix for covariance & see Proposition \ref{prop: finite CLT}\\
        \hline
        $b$ & $1$ & Maximum drug effect parameter & $b > 0$\\

        $E$ & $1$ & Half maximum drug effect parameter & $E >0$\\

        $m$ & $1$ & Drug effect steepness parameter & $n>0$\\
        \hline
        $n$ & $1$ &  Initial total cell count & $n \in [1000,10000]$\\
        
        $c$ & $1$ & Standard deviation of the i.i.d. observation noise & $c\in (0,0.1n)$\\

        $\mathbf{p}$ & $1\times K$ & Vector of initial proportion $p_i$ for subtype $i$ & $\sum_{i = 1}^{K} p_i = 1, n_i = np_i$\\
        
        $\bm{\pi}(d)$ & $1\times K$ & \textit{Stable proportion} & $\bm{\pi}(d)\mathbf{A}(d) = \lambda_1(d) \bm{\pi}(d)$\\
        
        $\mathcal{D}$ & $1\times N_{D}$ & Set of concentration levels applied & $|\mathcal{D}| = N_D$\\

        $\mathcal{T} $ & $1\times N_T$ & Set of time points & $|\mathcal{T}| = N_T$\\

        $\mathbf{W}_{n_i}^{(i)}(\mathcal{T},d)$ & $1\times N_T$ & Centered and normalized process of total cell counts & $\mathbf{W}_{n_i}^{(i)} (\mathcal{T},d) = \frac{1}{\sqrt{n_i}} \sum_{j = 1}^{n_i} \left(\mathbf{y}^{(i)}_{j}(\mathcal{T},d) - \bm{\mu}^{(i)}(\mathcal{T},d) \right)$\\

        $\bm{\mu}^{(i)}(\mathcal{T},d)$ & $1\times N_T$ & Expected value of the total cell counts & $\mu^{(i)}(\mathcal{T},d)\in \mathbb{R}^{1\times N_T}$\\

        $\mathbf{V}^{(i)}(\mathcal{T},d)$ & $N_T\times N_T$ & Covariance matrix of $\mathbf{W}_{n_i}^{(i)}(\mathcal{T},d)$ & $\mathbf{V}^{(i)}(\mathcal{T},d) \in \mathbb{R}^{N_T \times N_T}$\\



        \hline
        
    \end{tabular}}
    \caption{Table of definitions in Section \ref{sec:Cell dynamic model}, Section \ref{sec:drug-effect model}, Section \ref{sec: Long-run behavior}, and Section \ref{sec:Statistical model}.}
    \label{tab:notions}
\end{table}
}

\subsubsection{Deterministic approximation and time-delayed drug-effect model}

\label{sec: Deterministic approximation and time-delayed drug-effect model}

The statistical model \eqref{eq: stat model} utilizes the central limit theorem approximation derived in Proposition \ref{prop: finite CLT}. Alternatively, one can employ a law of large numbers type approximation to derive a simpler statistical model.
Specifically, we can formulate the statistical model as follows:
\begin{equation}
    \label{eq: det stat model}
    \mathbf{x}_{\mathcal{T},d,r} \approx \sum_{i = 1}^{K} n_i \bm{\mu}^{(i)}(\mathcal{T},d) + N^{(r)}(0,c^2 \mathbf{I}),
\end{equation}
where the variability within the data is solely attributed to the observation noise with a standard deviation of $c$.
We note that this model assumes deterministic cell growth dynamics, as $\mathbf{\mu}^{(i)}(\mathcal{T},d)$ is a deterministic function of time and dosage. 
In \cite{greene2019mathematical}, Greene et.al.~have studied a similar model; however, they did not specify the dose-response relationship using the nonlinear Hill equation and instead suggested a more restrictive linear relationship. 

The deterministic model allows for easier computation of the likelihood function; we therefore use this model when analyzing data with a time-delayed drug response.
Specifically, we define the time-inhomogeneous mean behavior of the asymmetrical division multi-type branching process as:
\[\Tilde{m}^{(i)}(t,d) = \expect*{\mathbf{B}^{(i)}(t,d)} = \mathbf{e}_i \exp\left(\int_{0}^{t}\mathbf{A}(\tau,d) d\tau\right),\]
where the \textit{infinitesimal generator matrix} $\mathbf{A}(t,d)$ depends on both time and concentration levels, as described in the logistic time-delayed drug effect model proposed in Section \ref{sec:time-delayed model}.
Although the corresponding covariance for $\mathbf{B}^{(i)}(t,d)$ can be computed in a similar manner, calculating the variance within a realistic time frame is challenging.
Therefore, we employ the simplified statistical model \eqref{eq: det stat model} with $\bm{\mu}^{(i)}(\mathcal{T},d) := [\langle \Tilde{m}^{(i)}(t_1,d),\mathbf{1}\rangle, \cdots, \langle \Tilde{m}^{(i)}(t_{N_T},d),\mathbf{1}\rangle ]$ when incorporating the logistic time-delayed drug effect later in Section \ref{sec:In vitro}.


\subsection{Simplified model of drug effect on CSCs and CNSCs mixture}

\label{sec:CSC assumptions}
In our computational analysis of Sections \ref{sec:In silico} and \ref{sec:In vitro}, we will investigate the performance of our newly proposed framework on a tumor population consisting of two phenotypes: 1. CSCs, 2. CNSCs. Our primary focus is on identifying drug-induced plasticity, that is drug-induced transitions from the CNSC phenotype to the CSC phenotype. Given that we only have two subpopulations, we identify the parameters for the CSC resistance subpopulation with an $r$ subscript and those for the CNSC sensitive subpopulation with a $s$ subscript, as shown in Table \ref{tab:parameters table}:

\begin{table}[ht]
    \centering
    \begin{tabular}{|c|c|c|c|c|c|c|c|c|}
    \hline
    Cell type & Initial proportion & $\alpha$ rate & $\beta$ rate & $\nu$ rate & $b$ for $\beta$ & $E$ for $\beta$ & $b$ for $\nu$ & $E$ for $\nu$\\
    \hline
    CSCs & $p_{r}$& $\alpha_r$ & $\beta_r$ & $\nu_{rs}$ & $b_{r,\beta}$ & $E_{r,\beta}$ & $b_{r,\nu}$ & $E_{r,\nu}$\\
    \hline
    CNSCs& $p_{s}$& $\alpha_s$ & $\beta_s$ & $\nu_{sr}$ & $b_{s,\beta}$ & $E_{s,\beta}$ & $b_{s,\nu}$ & $E_{s,\nu}$ \\
    \hline
    \end{tabular}
    \caption{CSCs and CNSCs subpopulation parameters}
    \label{tab:parameters table}
\end{table}

To investigate heterogeneous drug effects on CSCs and CNSCs, we adopt assumptions inspired by the experimental work \cite{padua2023high}. First, we specify assumptions related to the cell growth dynamics:

\begin{assumption}
    \label{assump:Cell dynamic}
    CNSCs proliferate with less potential than CSCs, while CSCs have a non-negative natural proliferation rate and a non-negative differentiation rate, expressed as $\kappa_r \geq \kappa_s \geq 0, \nu_{rs} \geq 0$.
\end{assumption}

\begin{assumption}
    \label{assump:Cell plasticity}
    CNSCs do not exhibit natural plasticity, implying no transition from CNSCs to CSCs in the absence of drug.
\end{assumption}
\noindent
We note that under Assumption \ref{assump:Cell plasticity}, the irreducibility condition of Section \ref{sec:Cell dynamic model} is violated in the absence of drug. However our framework can still be applied under the following non-extinction assumption, as discussed further in Appendix \ref{appx: Treatment under potential violation}.


\begin{assumption}
    \label{assump:no extinct CSCs}
    CSCs do not go extinct, ensuring the persistence of the CSC subpopulation.
\end{assumption}


Additionally, we constrain the pharmacologic dynamics of cells as follows:
\begin{assumption}
    \label{assump:Drug effect}
    The drug does not affect CSCs, reflecting a resistant behavior in CSCs.
\end{assumption}
\noindent
We note that if the starting population of CSCs is sufficiently large and CSCs have a positive net growth rate, $\kappa_r>0$, then Assumption 3 will follow from Assumption 4, since a large starting population of proliferating cells unaffected by the drug is unlikely to go extinct.

The final assumption pertains to the initial distribution between CSCs and CNSCs. 
In \cite{padua2023high}, the researchers experimentally separated CSCs and CNSCs, and they performed assays on isolated subpopulations of CNSCs. 
However, assuming the \textit{stable proportion} between CSCs and CNSCs at the start (Section \ref{sec: Long-run behavior}) -- without employing a separation technique -- is more natural.
Hence, we begin with Assumption \ref{assump:Initial structure} and later relax it in Section \ref{sec:relaxed ip assumption}.

\begin{assumption}
    \label{assump:Initial structure}
    The initial proportion of CSCs and CNSCs in the absence of the drug is the stable proportion.
\end{assumption}
\noindent
By adopting this assumption, there is no need to estimate the initial proportions $p_r$ and $p_s$; these parameters are inferred directly from the cell dynamic parameters $(\alpha,\beta,\nu)$ for each subpopulation.


The corresponding mathematical formulations for the above assumptions are summarized in Table \ref{tab:assumption table}.

\begin{table}[ht]
    \centering
    \begin{tabular}{|c|c|}
        \hline
        Assumptions & Reflection in parameter \\
        \hline
        Assumption \ref{assump:Cell dynamic} &  $\kappa_r \geq \kappa_s \geq 0 , \nu_{sd} \geq 0$\\
        \hline
        Assumption \ref{assump:Cell plasticity} & $\nu_{sr} = 0$\\
        \hline
        Assumption \ref{assump:Drug effect} & $b_{r,\beta} = b_{r,\nu} = 1$\\
        \hline
        Assumption \ref{assump:Initial structure} & $\mathbf{p} = \bm{\pi}(0)$\\
        \hline
    \end{tabular}
    \caption{Assumptions and corresponding parameter constraints}
    \label{tab:assumption table}
\end{table}


\section{Results (\textit{in silico})}

\label{sec:In silico}

To evaluate our framework's performance in analyzing \textit{in silico} data, we utilized the Gillespie algorithm \cite{gillespie1976general} to generate computer-simulated data based on the asymmetrical birth multi-type branching process model.
Each simulation started with an initial total cell count of $n = 1000$, and we conducted $N_R = 20$ replicates across various concentration levels $\mathcal{D}$ and time points $\mathcal{T}$, where
    \[\mathcal{D} = \{0,0.0313,0.0625,0.125,0.25,0.375,0.5,1.25,2.5,3.75,5\},\]
    \[\mathcal{T} = \{0,3,6,9,12,15,18,21,24,27,30,33,36\}.\]
Finally, Gaussian observation noise was added to obtain the \textit{in silico} dataset. Further details on the data generation can be found in Appendix \ref{appx: data generation and optimization implementation.}.

In these \textit{in silico} experiments, we derive two types of estimation results: 1. point estimation (PE) obtained for each experiment through the MLE process (see Appendix \ref{appx: data generation and optimization implementation.} for implementation details), 2. confidence intervals (CIs) obtained using the bootstrapping technique (see Appendix \ref{appx: bootstrapping} for implementation details), specifically for selected \textit{in silico} experiments.

Given our focus on a specialized case involving two subpopulations, as outlined in Section \ref{sec:CSC assumptions}, we set $K =2$. It is worth noting that one can treat $K$ as a variable and employ model selection techniques to estimate the number of subpopulations exhibiting distinct phenotypic drug responses. Similar approaches were explored in \cite{kohn2023phenotypic} using a simpler statistical framework, but a detailed exploration of this topic is outside the scope of the current study.


\subsection{Deconvolution of drug effects (base experiment)}


\subsubsection{ Illustrative example}

\label{sec:Illustrative example}

We begin with a single illustrative example. Table \ref{tab:True parameter I1} shows the true parameter set used to generate the {\em in silico} dataset and the corresponding PEs obtained using our statistical framework. 
In Figure \ref{fig:illustrative}, we display both the PEs and CIs for this example, focusing on two key metrics: the \textit{stable proportion} $\bm{\pi}(d)$ between phenotypes under each concentration level $d$, and the $GR_{50}$ dose for the CNSC growth rate (cytotoxic effect) and de-differentiation rate (plasticity effect), respectively.
The $GR_{50}$ dose is the concentration level at which the drug achieves half its maximal observed effect on the growth rate.
It is a useful summary metric of the drug effect which incorporates both primary drug effect parameters $E$ and $b$, as is further discussed in Appendix \ref{appx: Parameters of interest}.

\begin{table}[ht]
    \centering

    \begin{tabular}{|c|c|c|c|c|c|c|c|c|c|}
        \hline
        & $\alpha_r$ & $\beta_r$ & $\nu_{rs}$ & $\alpha_s$ & $\beta_s$  & $b_{s,\beta}$ & $E_{s,\beta}$ & $b_{s,\nu}$ & $E_{s,\nu}$\\
        \hline
        $\btheta^*$ & 0.5407 & 0.5055 & 0.2929 & 0.2280 & 0.2280 &  0.8536 & 0.7073 & 1.0827 & 1.2285\\
        \hline
        $\hat{\btheta}$ & 0.5407 & 0.5036 & 0.2988 & 0.2412 & 0.2419  & 0.8612 & 0.6468 & 1.0773 & 1.2336\\
        \hline
    \end{tabular}
    \caption{True parameter set $\btheta^*$ and point estimates $\hat{\btheta}$ in Figure \ref{fig:illustrative}}
    \label{tab:True parameter I1}
\end{table}

\begin{figure}[ht]
    \centering
    \includegraphics[width = \linewidth]{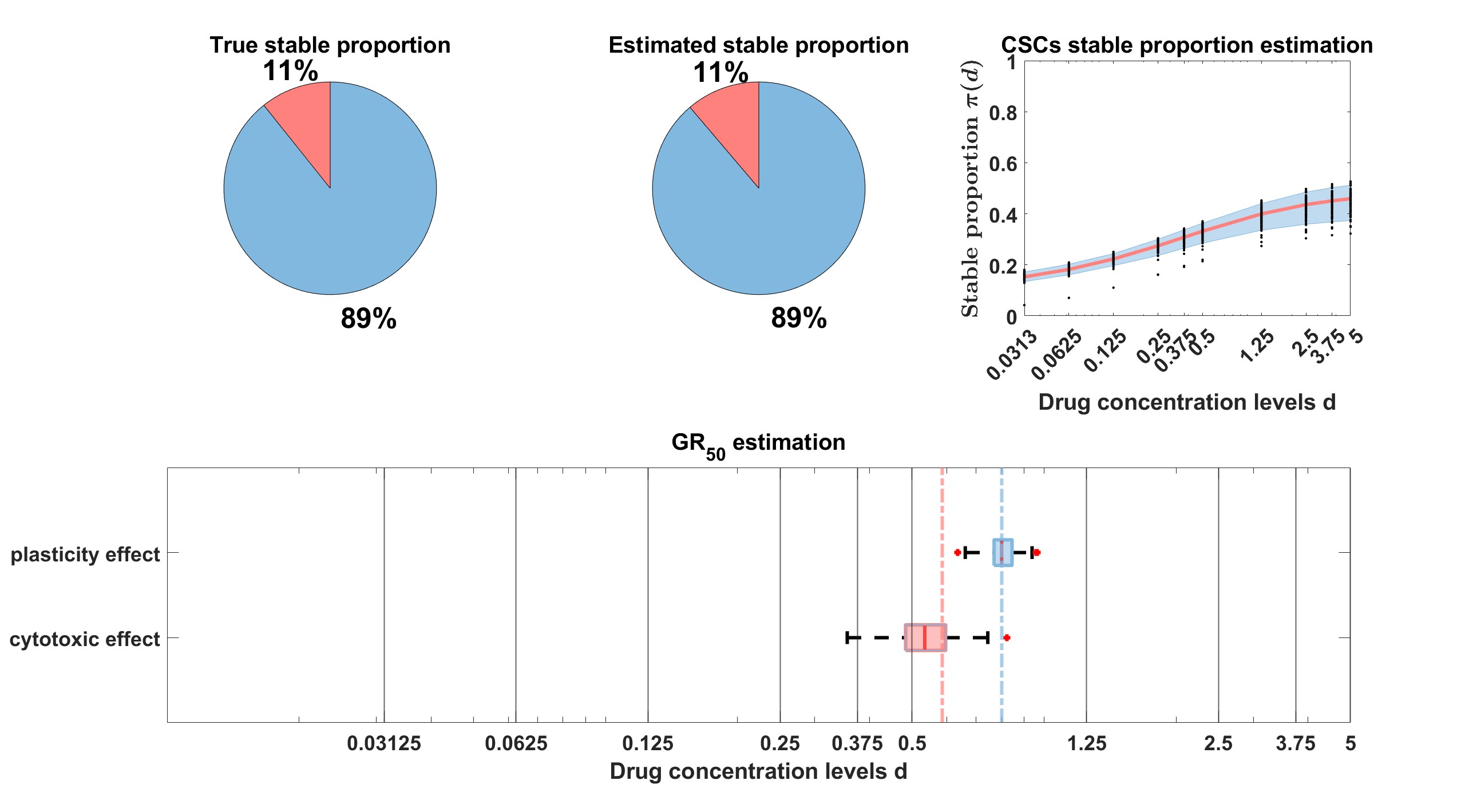}
    \caption{Estimation of \textit{stable proportion} and $GR_{50}$ based on the data generated by the true parameter set in Table \ref{tab:True parameter I1}. The pie chart demonstrates the PE of the \textit{stable proportion} between CSCs and CNSCs under a no-drug environment. The upper-right plot shows 100 bootstrapped PEs as scatter points and the corresponding $90\%$ CIs for the \textit{stable proportion} of CSCs at each drug concentration level, with the red line representing the true \textit{stable proportion} of CSCs in the same plot. The boxplots illustrate the CIs of $GR_{50}$ of both the drug-induced plasticity effect and cytotoxic effect, with the dashed colored lines representing the corresponding true $GR_{50}$ values. Specifically, the red dashed line and red boxplot correspond to cytotoxic effect, while the blue dashed line and blue boxplot correspond to drug-induced plasticity effect.}
    \label{fig:illustrative}
\end{figure}

Figure \ref{fig:illustrative} illustrates that our newly proposed framework accurately estimates the \textit{stable proportion} in the absence of drug (pie charts in top panel). 
Additionally, the varying \textit{stable proportions} of CSCs at different drug concentrations are captured by the CIs constructed from the bootstrapping estimations (line plot in top panel).
The CIs also encompass the correct $GR_{50}$ values for both the cytotoxic and plasticity effects (bottom panel).

In this example, the true $GR_{50}$ values for the two drug effects fall between two concentration levels applied in the experiment. 
In our previous work \cite{wu2024using}, we used a similar statistical framework to investigate a tumor with two subpopulations, where each was affected by an anti-cancer drug to a varying degree and there were no transitions between the subpopulations.
There, we observed that close $GR_{50}$ values for the cytotoxic effects on each subpopulation often resulted in poor estimations of the $GR_{50}$ values. 
In the current study, where we consider two distinct types of drug effects, the difference in $GR_{50}$ is no longer critical for distinguishing these effects.


\subsubsection{Estimation across a wide range of true parameter sets}
\label{sec:wide range parameter sets}
To assess estimation performance across a wide range of biologically realistic parameter sets, we examined the relative errors (RE) from 100 datasets, each generated using a different true parameter set. The RE between an estimator $\hat{x}$ and the true value $x^*$ was calculated as
\begin{equation}
    Er(\hat{x};x^*) =  \left\lvert \frac{\hat{x}-x^*}{x^*}\right\rvert.
\end{equation}
The results for all datasets are summarized using boxplots in Figure \ref{fig:RE_100}. 
Note that we recorded the RE for $1 - b_{s,\beta}$ and $b_{s,\nu} - 1$ rather than $b_{s,\beta}$ and $b_{s,\nu}$ directly.
These new measurements of REs ensure symmetry when measuring the maximum drug effect between cytotoxic and plasticity effects, since $b_{s,\beta} \in (0.5,1)$ and $b_{s,\nu}\in (1,1.5)$ (see Appendix \ref{appx: data generation and optimization implementation.}).
Additionally, we included a horizontal line at RE $=0.2$ as a threshold to indicate reasonably well estimated parameters. 

\begin{figure}[ht]
    \centering
    \includegraphics[width = \textwidth]{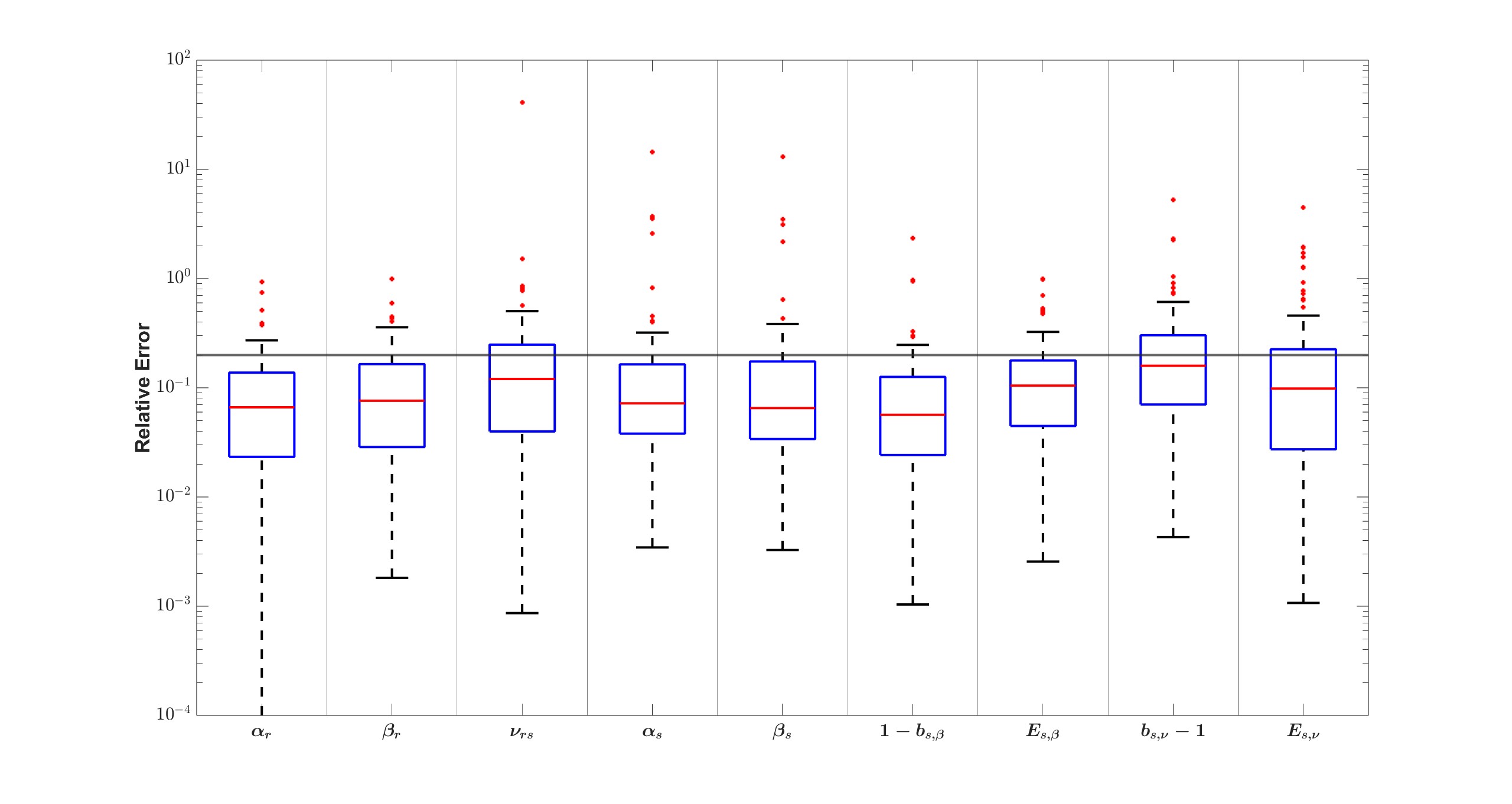}
    \caption{Relative error of the parameter estimation. The boxplot contains 100 independent experiments generated by randomly selected true parameter sets. The red line is the median of the REs. The grey horizontal line is a threshold when RE is equal to $0.2$.}
    \label{fig:RE_100}
\end{figure}

Figure \ref{fig:RE_100} shows that most point estimates are reasonable, with over three fourths of estimates having a RE below the specified threshold.
This indicates that our newly proposed framework can successfully disentangle the cell growth dynamics and drug effect dynamics. 
We note that the plasticity effect parameter $b_{s,\nu}$ is generally estimated less accurately than the cytotoxic effect parameter $b_{s,\beta}$. 
This may stem from the fact that in all parameter sets, the true plasticity effect is smaller than the cytotoxic effect, i.e. $b_{s,\nu} - 1 < 1 - b_{s,\beta}$, making it more challenging to detect the plasticity effect (see Appendix \ref{appx: data generation and optimization implementation.}).
In addition, we note that the relative error metric is more sensitive when the true value is small, since even small absolute errors $\lvert\hat{x}-x^*\rvert$ can lead to large relative errors $\lvert(\hat{x}-x^*)/x^*\rvert$.

Although our estimation framework performs well overall, there are some examples where it fails to recover the true values from the data.
One explanation for low quality in the estimation is a poor experimental design, i.e.~poor selection of concentration levels and data collection time points, resulting in reduced parameter identifiability.
For instance, when estimating $E$, if the tested concentration levels do not cover the true value, i.e. $\max(\mathcal{D}) < E$, it will be challenging for our framework to identify the true value.
Even when $E$ is within the tested concentration levels, observation noise and the stochastic nature of the data generation can still distort the estimation, particularly when $E$ is close to the maximum tested concentration level.
A detailed analysis of how experimental design impacts parameter identifiability will be explored in future work.

\subsubsection{Analysis of failed estimations}

As previously mentioned, one potential explanation for poor estimation under our framework is a poor experimental design.
To better understand potential challenges in parameter identifiability, we analyzed the true parameter sets and the simulated data from the 100 experiments described earlier. 
We specifically examined the drug effects on two theoretical metrics of long-run behavior: \textit{stable proportion} ($\bm{\pi}(d)$) and \textit{long-run growth rate} ($\lambda_1(d)$) (see Section \ref{sec: Long-run behavior}). 
For each metric, we computed its theoretical values across different concentration levels and determined the maximum discrepancy among them to quantify the ranges of drug effects observed in the experiments.

In Figure \ref{fig:Drug effect vs Average RE}, scatter plots illustrate these quantities plotted against the average RE over all estimated parameters for each experiment.
In the right panel's lower right corner, four examples show a maximum change in stable proportion below $0.1$, coinciding with a notably high average estimation error.
In the left panel, these same four examples exhibit low maximum changes in long-run growth rates.
This observation indicates that inaccurate estimations tend to occur when observed drug effects are minimal.
This insight is further supported by Figure \ref{fig:data example}, where we visualized two example datasets with high vs.~low RE, and observed that the  dataset with less variation in growth pattern across different concentration levels has higher average RE.
To summarize, poor estimation in terms of RE can in most cases be traced to a small observed drug effect, which can either arise due to a poor experimental design (drug effects do not manifest due to poor selection of concentration levels or data collection points) or a small true drug effect, in which case the RE metric is more sensitive to estimation inaccuracy as mentioned in the previous section.

\begin{figure}
    \centering
    \includegraphics[width = \textwidth]{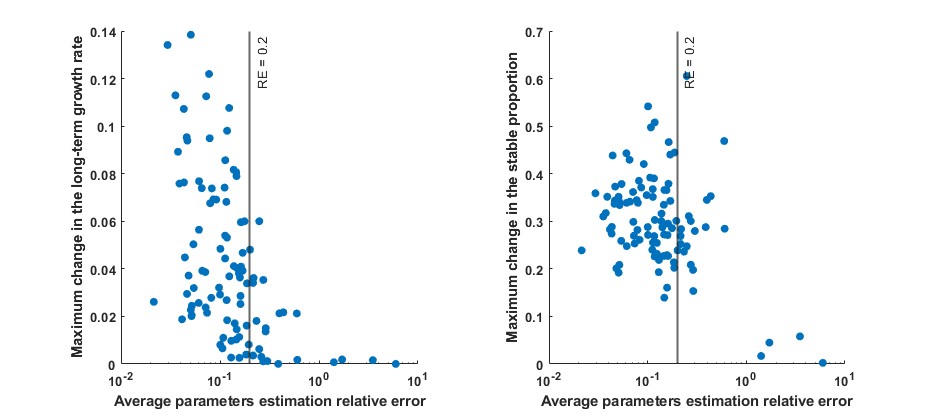}
    \caption{Scatter plots of 100 independent experiments. Both plots have average parameters estimation relative error on the x-axis. In the left panel, the y-axis represents the observable drug effect on the long-term growth rate of the CSCs and CNSCs mixture. In the right panel, the y-axis represents the observable drug effect on the \textit{stable proportion} between the CSCs and CNSCs.}
    \label{fig:Drug effect vs Average RE}
\end{figure}

\begin{figure}
    \centering
    \includegraphics[width = \textwidth]{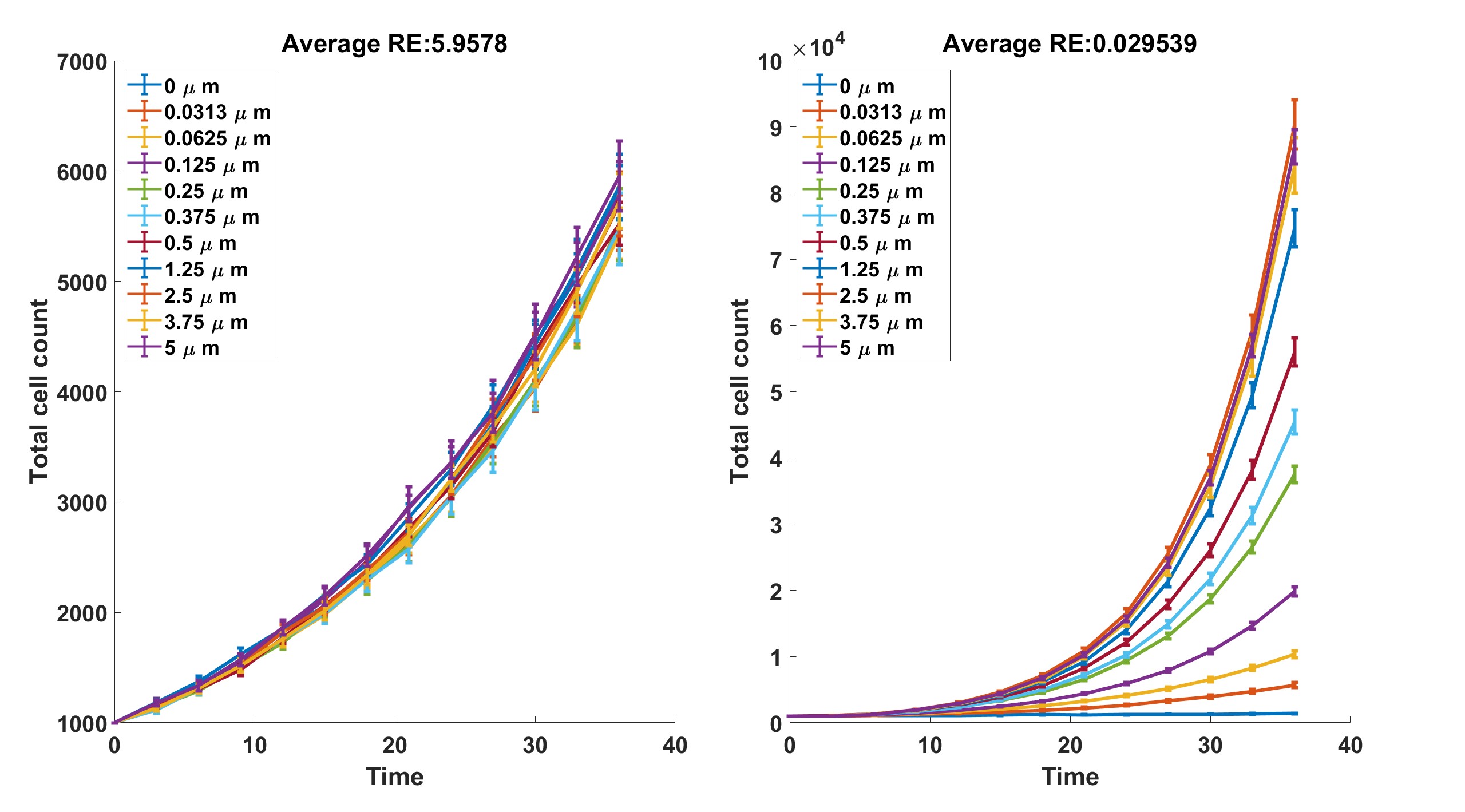}
    \caption{Visualization of the simulated data with sufficient and insufficient information for accurate estimation. The left panel demonstrates data with insufficient information, while the right panel demonstrates data with sufficient information. Each line plot illustrates total cell count data across 13 time points under a specific concentration level described in the legend. The error bar is derived from the standard deviation of 20 independent replicates.}
    \label{fig:data example}
\end{figure}





\subsection{Robustness of the framework}

To assess the robustness of our newly proposed statistical framework using \textit{in silico} data, we designed a series of experiments to evaluate the performance of the framework when relaxing some of the assumptions in Section \ref{sec:CSC assumptions}. 
Detailed experimental settings are described in Appendix \ref{appx: data generation and optimization implementation.}.

\subsubsection{Relaxed drug effect assumption}

\label{sec:relaxed drug assumption}

In Assumption \ref{assump:Drug effect} of Section \ref{sec:CSC assumptions}, we posit that the drug will not affect the CSCs. 
However, this assertion may seem too stringent for practical scenarios, where the varying microenvironment caused by different drug concentration levels could potentially affect the CSCs \cite{yang2020targeting}. 
Therefore, we conduct another set of experiments to examine the performance of our framework when the drug does influence the CSCs.
While the mechanisms underlying drug resistance in CSCs are not yet fully understood \cite{rezayatmand2022drug,li2021drug}, for testing purposes, we make the simplifying assumption that the drug affects CSCs in a similar way to CNSCs, manifesting through cytotoxic effects and increased asymmetrical divisions, albeit to a lesser extent than for CNSCs.
As noted in Section \ref{sec:drug-effect model}, our framework is also adaptable to scenarios where the drug might decrease the differentiation rate.


\begin{figure}
    \centering
    \includegraphics[width = \textwidth]{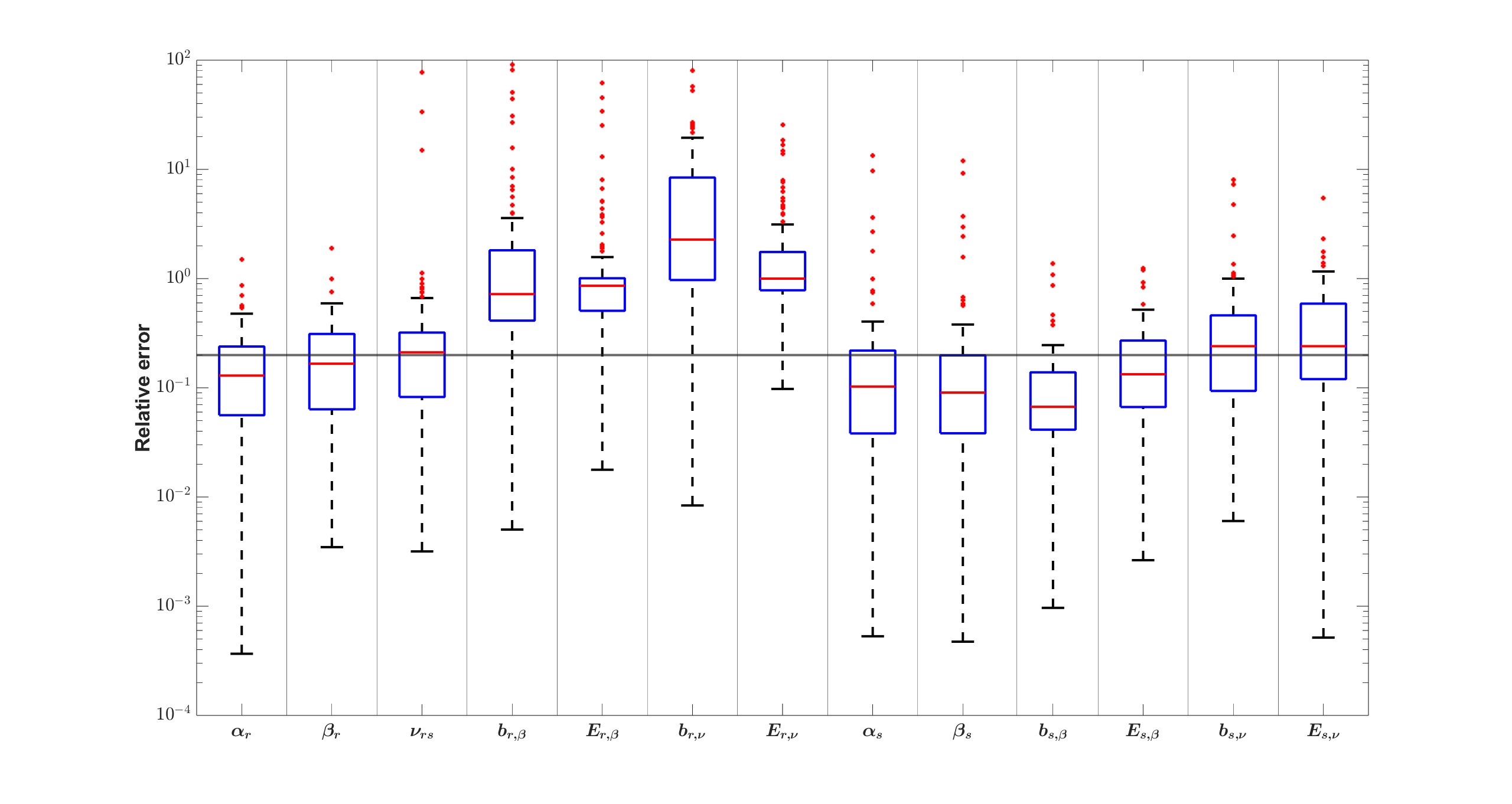}
    \caption{Relative error of the parameter estimation when relaxing Assumption \ref{assump:Drug effect}. The remaining details of the plot can be found in Figure \ref{fig:RE_100}.}
    \label{fig:Relaxed drug effect RE 100}
\end{figure}

Figure \ref{fig:Relaxed drug effect RE 100} presents the RE from 100 independent experiments. 
We observe that the estimates of certain parameters, such as  $(\nu_{rs}, b_{s,\nu}, E_{s,\nu})$, show deterioration compared to Figure \ref{fig:RE_100}, with more than $50\%$ of estimates having RE above the $0.2$ threshold value. 
However, our framework continues to provide reasonably accurate estimates for parameters like $\alpha_s,\beta_s,b_{s,\beta}$.

As mentioned in the previous section, the estimation accuracy may be correlated to the significance of the drug effect. 
In the current experiment, we find more compelling evidence to support this observation. 
In particular, the estimation accuracy for the CSC drug effect parameters $(b_{r,\beta},E_{r,\beta},b_{r,\nu},E_{r,\nu})$ is significantly lower than the accuracy in estimating the analogous parameters for CNSCs, due to our assumption of a more pronounced impact of the drug on CNSCs.

\subsubsection{Relaxed initial proportion assumption}

\label{sec:relaxed ip assumption}

Our next experiment aims to investigate the possibility of estimating the initial subpopulation structure under arbitrary initial proportions, relaxing Assumption \ref{assump:Initial structure} of Section \ref{sec:CSC assumptions}.
To accommodate this situation, we reintroduce two parameters, $p_r$ and $p_s$, with the constraint that $p_r + p_s = 1$, representing the initial proportions of CSCs and CNSCs respectively.

\begin{figure}
    \centering
    \includegraphics[width = \textwidth]{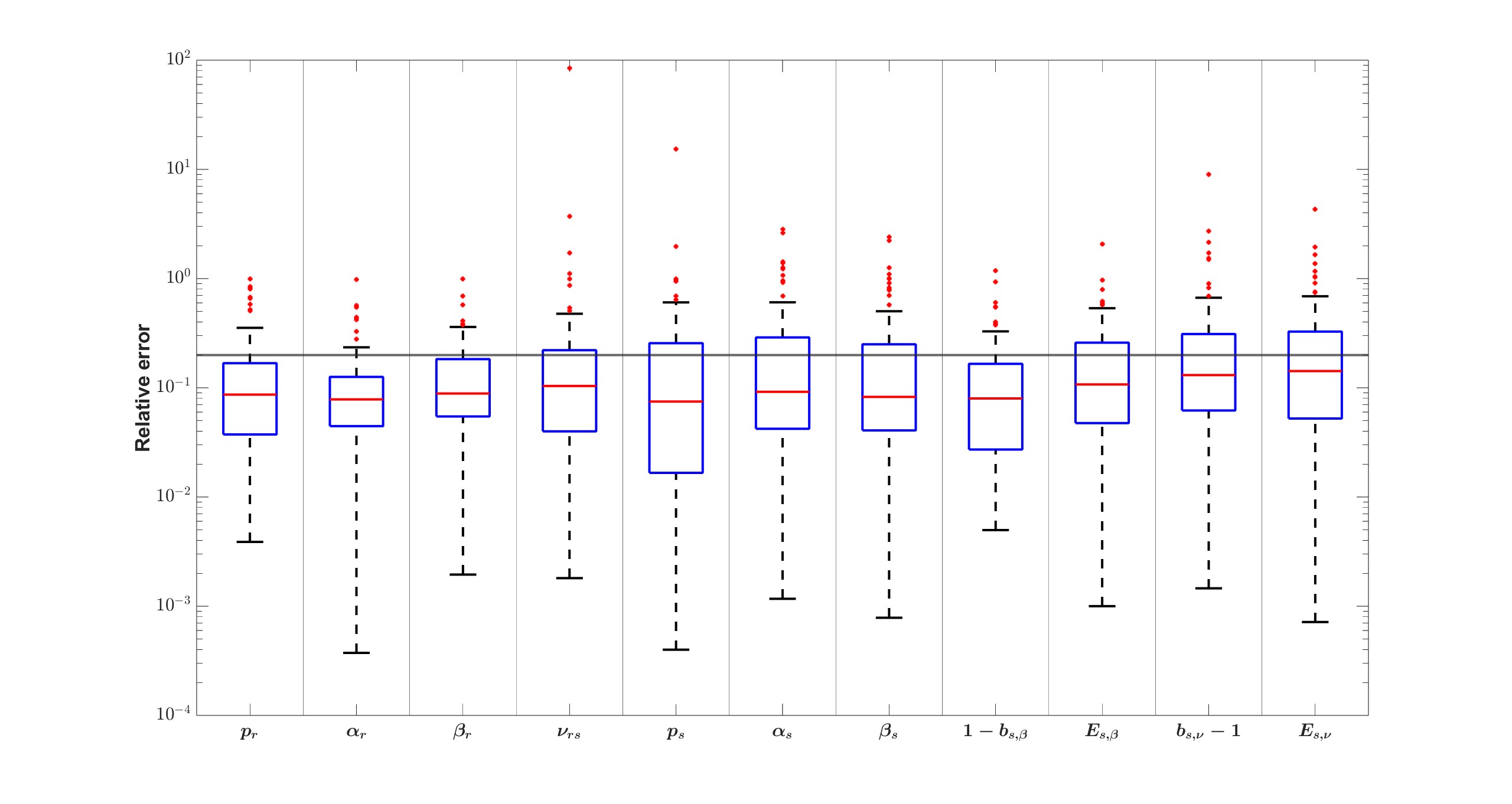}
    \caption{Relative error of the parameter estimation when relaxing Assumption \ref{assump:Initial structure}. The remaining details of the plot can be found in Figure \ref{fig:RE_100}.}
    \label{fig:Relaxed initial proportion RE 100}
\end{figure}

The results in Figure \ref{fig:Relaxed initial proportion RE 100} indicate that relaxing the initial proportion assumption does not deteriorate the accuracy of the estimation of other parameters, and the initial proportion is estimated reasonably well. 
This highlights the potential of our framework to deconvolute the phenotypic subpopulation structure with an unknown initial proportion. 
It is also worth noting that we allow the $p_r$ and $p_s$ to be uneven, for example, $p_r = 0.03$ and $p_s = 0.97$.

\subsubsection{Limited division CNSC dynamics assumption}

\label{sec:limited division cell growth dynamic assumption}

We finally consider a more general CNSC division scheme, where CNSCs can only undergo symmetric division a limited number of times, reflecting a gradual loss of proliferation potential.
We refer to CNSCs produced through asymmetric division of CSCs as first-generation CNSCs, which can only divide up to $G$ times. 
For the first $G$ generations of CNSCs, we assume that the $g$-th generation can symmetrically divide into two $(g+1)$-th generation CNSCs, asymmetrically divide into one CSC and one $g$-th generation CNSC, or die. 
The $(G+1)$-th generation CNSCs can only die or asymmetrically divide into one CSC and one $(G+1)$-th generation CNSC. 
All CNSCs share the same symmetric division rate $\alpha_s$, death rate $\beta_s$ and asymmetric division rate $\nu_{sr}$ (Figure \ref{fig:Limited division events}).
In addition, the drug effect parameters $b_{s,\beta},E_{s,\beta}, b_{s,\nu}, E_{s,\nu}$ are the same for CNSCs in the different generations.
This scheme encompasses a broader spectrum of CNSC dynamics, ranging from the scenario where CNSCs cannot self-renew ($G = 0$) to the scenario where CNSCs have an infinite capacity for self-renewal with a certain rate $\alpha_s$ ($G \rightarrow \infty$).
Theoretically, our current framework can accurately capture the extreme scenarios where $G = 0$ or $G \rightarrow \infty$.
We are interested in the practical performance of our framework in the intermediate case, i.e., $0<G <\infty$.

\begin{figure}[ht]
\begin{subfigure}{.23\textwidth}
  \centering
  \includegraphics{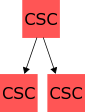}  
  \caption{Symmetric division of CSC (rate: $\alpha_r$)}
  \label{fig:CSC_sd}
\end{subfigure}
\begin{subfigure}{.23\textwidth}
  \centering
  \includegraphics{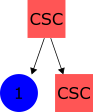}  
  \caption{Asymmetric division of CSC (rate: $\nu_{rs}$)}
  \label{fig:CSC_ad}
\end{subfigure}
\begin{subfigure}{.23\textwidth}
  \centering
  \includegraphics{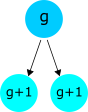}
  \caption{Symmetric division of CNSC (rate: $\alpha_s$)}
  \label{fig:CNSC_sd}
\end{subfigure}
\begin{subfigure}{.23\textwidth}
  \centering
  \includegraphics{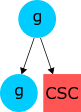}
  \caption{Asymmetric division of CNSC (rate: $\nu_{sr}$)}
  \label{fig:CNSC_ad}
\end{subfigure}
\caption{Illustration of four distinct division events assumed for the limited division scenario. Square shape represents CSC and round shape indicates CNSC, while different colors illustrate a different division potential of each generation in CNSC.}
\label{fig:Limited division events}
\end{figure}
 



We first want to determine whether our framework can reliably estimate the symmetric division rate of CNSCs, $\alpha_s$.
To streamline the experiment, we selected the same true value $\alpha_s^* > 0$ for $G\in\{0,1,2,3\}$.
Instead of using the RE, we adopt the ratio between the estimated value and the true value of $\alpha_s$ as the metric of accuracy. 
Specifically, for an estimator $\hat{x}$ and the true value $x^*$, we compute the accuracy metric as:
\begin{equation}
    \label{eq:ratio error}
    Er(\hat{x};x^*) = \frac{\hat{x}}{x^*}.
\end{equation}
This new metric effectively captures two important scenarios: (i) the framework should produce an estimation $\hat{\alpha}_s = 0$ when $G = 0$, resulting in a ratio $\hat{\alpha}_s/\alpha_s^* = 0$, and (ii) the framework should accurately estimate $\hat{\alpha}_s = \alpha_s^*$ when $G$ is sufficiently large.

\begin{figure}
    \centering
    \includegraphics[width = \textwidth]{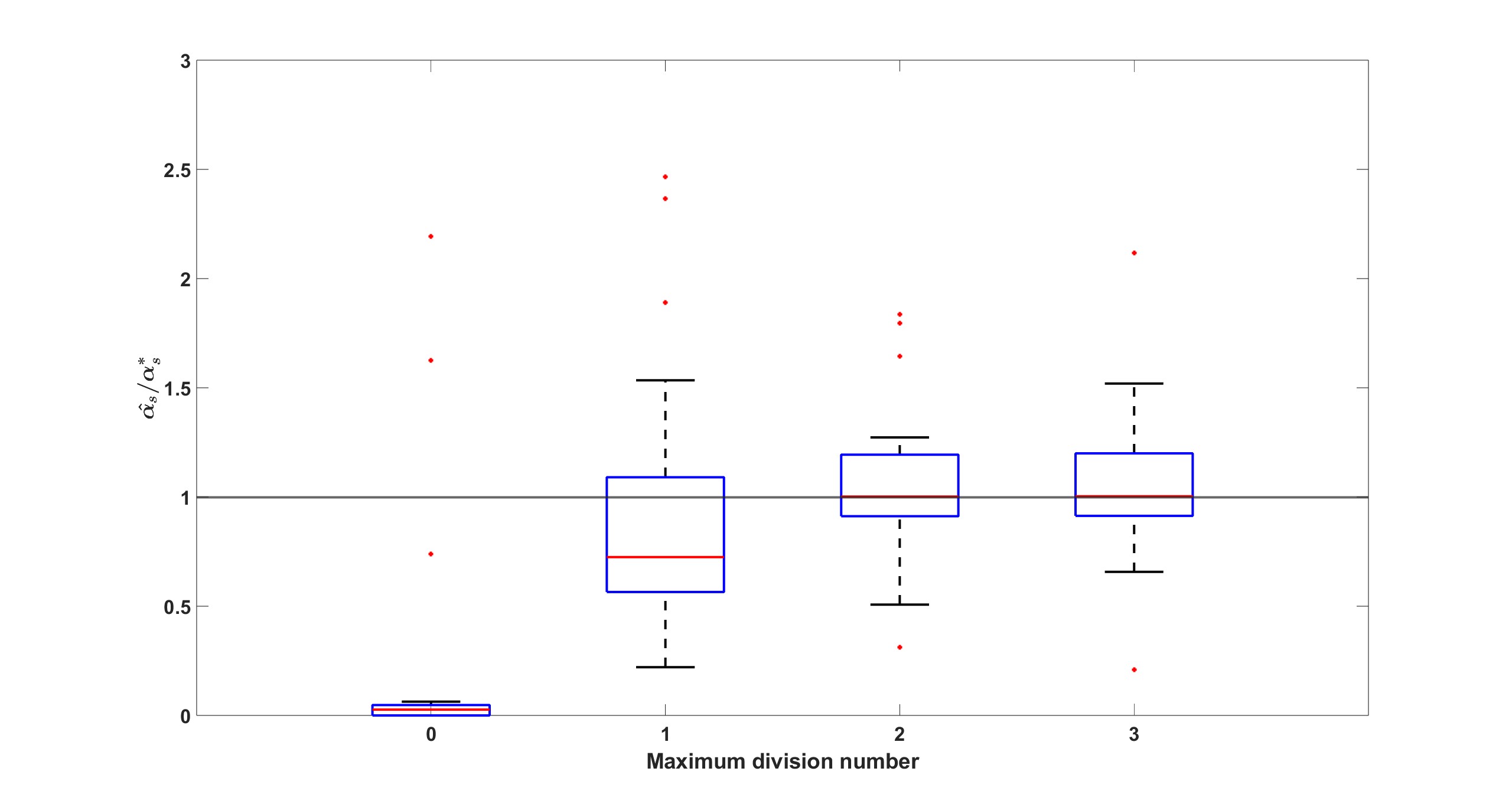}
    \caption{Estimation error of the CNSCs division rate when CNSCs can only divide a limited amount of times. The ratios between the estimated division rate and the true division rate are presented to varying maximum division numbers $G$. The boxplot is based on 30 independent experiments for each $G\in \{0,1,2,3\}$ with varying true parameter sets. The grey horizontal line illustrates the ratio when CNSCs division rate is accurately estimated.}
    \label{fig:birth rate ratio}
\end{figure}

The results of our experiment (Figure \ref{fig:birth rate ratio}) align well with these two scenarios.
Specifically,  our framework accurately infers that $\alpha_d = 0$ when the maximum division number is $G = 0$, indicating that CNSCs cannot undergo self-renewal. 
However, when $G = 1$, $75\%$ of estimates for $\alpha_s$ are below the true value.
Surprisingly, our framework accurately recovers the true value of $\alpha_s$ when $G = 2$ and $G = 3$.
For other parameters, we generally observed systematic biases in the estimation, see Appendix \ref{appx: detailed analysis of limited division CNSCs dynamic assumption}. 
However, interestingly, we found that the drug effects' inflection points, $E_{s,\nu}$ and $E_{s,\beta}$, were accurately estimated across different values of $G$. 
This observation demonstrates the robustness of our framework in estimating the drug effects' inflection points under varying differentiation dynamics.
In other words, through detecting the changes in cell population growth dynamics, our framework can estimate the concentration levels at which such changes occur, even though it may not fully explain these changes.
For the other parameters, we note that it is possible to fully capture the dynamics of CNSCs with limited proliferation potential by expanding our model beyond two subpopulations, as we further discuss in the conclusion section (Section \ref{sec:Conclusion}).




\section{Results (\textit{in vitro})}

\label{sec:In vitro}
We validate our approach using several published drug screen datasets. 
\subsection{Detecting CPX-O induced plasticity in gastric carcinoma cell lines}

\subsubsection{Data overview}

To validate our framework using \textit{in vitro} data, we utilized data from a recent study \cite{padua2023high} that investigated four human gastric carcinoma cell lines (AGS SORE6$-$, AGS SORE6+, Kato III SORE6$-$, Kato III SORE6+). 
Our focus in this section is specifically on the AGS SORE6$-$ and AGS SORE6+ cell lines. 
SORE6 serves as a reporter system indicating the stemness of these cancerous gastric cells, where SORE6$-$ and SORE6+ denote gastric CNSCs and CSCs, respectively. 
In the study \cite{padua2023high}, ciclopirox olamine (CPX-O) was found to reprogram gastric CNSCs into CSCs. 
To support these findings, the authors cultivated pure AGS SORE6$-$ cells and monitored both cell viability data and the percentage of AGS SORE6+ cells using the SORE6 reporter system.

In our study, we mainly utilized two experimental datasets provided in \cite{padua2023high}: 
\begin{enumerate}
    \item \textit{AGS Conc $(X_{Conc})$:} AGS SORE6$-$ samples were cultivated under nine concentration levels of CPX-O: 
    \[\mathcal{D}_1 = \{0,0.125,0.25,0.5,1,2,4,8,16\} \; (\mu M).\]
    Total cell count data (TC) and stem cell proportion (SC) were observed after 48 hours of treatment, i.e. $\mathcal{T}_1 = \{48\} $ (hours).

    \item \textit{AGS Time $(X_{Time})$:} AGS SORE6$-$ samples were cultivated under three concentration levels of CPX-O:
    \[\mathcal{D}_2 = \{0,4,8\} \; (\mu M).\]
    Total cell count data (TC) and stem cell proportion (SC) were observed after 2, 6, 12, 24 and 48 hours of treatment: 
    \[\mathcal{T}_{2} = \{2,6,12,24,48\} \text{ (hours)}.\]
\end{enumerate}

Unfortunately, neither of these experiments provides enough information to estimate all parameters under our framework.
Specifically, three concentration levels in the \textit{AGS Time} data are not sufficient to estimate the four parameter effects of the drug in our framework, and a single time point in the \textit{AGS Conc} data is inadequate to capture the heterogeneous dynamics of cell growth.
This limitation may arise from the need to employ the FACS technique to distinguish between the CSC and CNSC populations. 
To overcome this limitation, we propose to employ a data imputation on the \textit{AGS Conc} data based on the information from the \textit{AGS Time} data. 
Further details of imputation can be found in Appendix \ref{appx: AGS Conc estimated TC}. We provide the imputed dataset in Table \ref{tab:AGS Conc data.}.

\begin{table}[ht]
    \centering
    \begin{tabular}{|c|c|c|c|c|c|c|c|c|c|}
    \hline
       Concentration levels ($\mu M$) & 0 & 0.125 & 0.25 & 0.5 & 1 & 2& 4 & 8 & 16  \\
       \hline
       TC of replicate 1 at time 12  & {\tb 3083} & {\tb 3083} & {\tb 3083} & {\tb 3083} & {\tb 3083} & {\tb 3083} & {\tb 3083} & {\tb 3083} & {\tb 3083}\\
       \hline
       TC of replicate 1 at time 24  & {\tb 4272} & {\tb 4246} & {\tb 4222} & {\tb 4173} & {\tb 4080} & {\tb 3907} & {\tb 3608} & {\tb 3150} & {\tb 2568}\\
       \hline
       TC of replicate 1 at time 48 & 5038 & 4503 & 3242 & 2235 & 1849 & 1608 & 1140 & 1054 & 756 \\
       \hline
       TC of replicate 2 at time 12  & {\tb 3502} & {\tb 3502} & {\tb 3502} & {\tb 3502} & {\tb 3502} & {\tb 3502} & {\tb 3502} & {\tb 3502} & {\tb 3502}\\
       \hline
       TC of replicate 2 at time 24  & {\tb 4851} & {\tb 4823} & {\tb 4795} & {\tb 4739} & {\tb 4633} & {\tb 4437} & {\tb 4098} & {\tb 3578} & {\tb 2917}\\
       \hline
       TC of replicate 2 at time 48 & 5722 & 4597 & 3433 & 2881 & 2511 & 2295 & 1875 & 1309 & 694 \\
       \hline
    \end{tabular}
    \caption{\textit{In vitro} total cell count (TC) data: \textit{AGS Conc} data with augmented estimation according to \textit{AGS Time}. The colored data are estimated data according to the \textit{AGS Time} data, while the black data are actual data from \cite{padua2023high}.}
    \label{tab:AGS Conc data.}
\end{table}

\subsubsection{Candidate model description and model selection results}

One potential application of our newly proposed statistical framework in practical settings is to detect the presence of drug-induced plasticity using model selection criteria. 
Here, we employ the well-known Akaike Information Criterion (AIC) calculated by the formula:
\begin{equation}
    \label{eq: AIC}
    \text{AIC} = 2|\btheta| - 2 \log(L^*),
\end{equation}
where $|\btheta|$ is the number of free parameters in the model and $L^*$ is the maximum value of the likelihood function for the model. 
AIC, rooted in information theory, quantifies the information lost by a model, with lower AIC values indicating higher model quality. 

We consider a total of four different models. 
The first two models, based on the framework described in equation \eqref{eq: stat model}, differ in their assumptions regarding the induction of plasticity by CPX-O. 
We denote the model that assumes drug-induced plasticity as model $I$ and the alternative as model $Ia$. 
In fitting the \textit{in vitro} data, we reintroduce the Hill parameter $m$ to our drug-effect model (Section \ref{sec:drug-effect model}). 
Additionally, it is assumed that the drug does not affect the dynamics of CSCs.

Model $I$ has $12$ free parameters, as detailed in Table \ref{tab:in vitro model parameters}, to depict the cell growth dynamics of CSCs and CNSCs and the drug response of CNSCs.
In contrast, Model $Ia$ requires only $6$ free parameters, as it assumes no drug-induced plasticity for CPX-O, resulting in no parameters necessary for gastric CSCs.
Indeed, since the {\em in vitro} experiments in question are started with isolated CNSC subpopulations $(p_s=1)$, and no natural plasticity of CNSCs is assumed ($\nu_{sr} = 0$), no CSCs should emerge in the experiments according to the assumptions of Model $Ia$.

In the other two models, we introduce a time-delayed drug-effect which is present in the \textit{AGS Time} data, since similar cell growth patterns are observed within the first 12 hours across different drug concentration levels, with pronounced variations emerging thereafter. 
Further details can be found in Appendix \ref{appx: AGS time delayed}. 
Therefore, we assume that the maximum drug effect parameter $b$ varies over time following a two-parameter logistic function as described in Section \ref{sec:time-delayed model}. 
Like the previous two models, one of these time-delayed drug-effect models, denoted as $II$, assumes that CPX-O can induce plasticity, while the other, denoted as $IIa$, does not. 
Model $II$ has $14$ free parameters, whereas Model $IIa$ involves $8$ free parameters, as detailed in Table \ref{tab:in vitro model parameters}.

\begin{table}
    \centering
    \begin{tabular}{|c|c|l|}
    \hline
        Model & $|\theta|$ & Free parameters  \\
        \hline
        $I$ & 12 & $\alpha_r,\beta_r,\nu_{rs},\alpha_s,\beta_d,b_{s,\beta},E_{s,\beta},m_{s,\beta},b_{s,\nu},E_{s,\nu},m_{s,\nu},c$  \\
        \hline
        $Ia$ & 6 & $\alpha_s,\beta_s,b_{s,\beta},E_{s,\beta},m_{s,\beta},c$ \\
        \hline
        $II$ & 14 & $\alpha_r,\beta_r,\nu_{rs},\alpha_s,\beta_s,b_{s,\beta},E_{s,\beta},m_{s,\beta},b_{s,\nu},E_{s,\nu},m_{s,\nu},k,t_0,c$  \\
        \hline
        $IIa$ & 8 & $\alpha_s,\beta_s,b_{s,\beta},E_{s,\beta},m_{s,\beta},k,t_0,c$  \\
    \hline
    \end{tabular}
    \caption{Free parameters in models for AGS--CPX-O experiment.}
    \label{tab:in vitro model parameters}
\end{table}

\begin{table}
    \centering
    \begin{tabular}{|c|c|c|}
    \hline
         & Drug-induced plasticity & no Drug-induced plasticity \\
         \hline
        Without time-delay effect & \textbf{AIC($I$): 653.4866} & AIC($Ia$): 677.9985\\
         \hline
        With time-delay effect & \textbf{AIC($II$): 543.7607}& AIC($IIa$): 576.5120\\
        \hline
    \end{tabular}
    \caption{AIC value for 4 varying model assumptions in AGS--CPX-O experiment.}
    \label{tab:AIC value}
\end{table}

The AIC results, summarized in Table \ref{tab:AIC value}, indicate a preference for the models assuming drug-induced plasticity under the AIC model selection criteria. 
Furthermore, comparing the AIC values between models with and without the time-delay assumption reveals the significant role of the time-delay effect in the current data. 
While our newly proposed framework may not fully capture the experimental data in \cite{padua2023high}, it nonetheless yields conclusions consistent with the actual observations reported in \cite{padua2023high}.
It should also be stressed that our framework is able to infer the presence of drug-induced resistance from total cell count data, whereas this conclusion is drawn in \cite{padua2023high} using additional data on the CSC proportion over time.
Of course, the inference of drug-induced plasticity is in this case facilitated by the ability to start the experiments from isolated CNSCs, but our framework is applicable to a broader range of experimental setups without this ability.

\subsection{Detecting Vemurafenib induced plasticity in COLO858 cell line}

In \cite{fallahi2017adaptive}, the authors observed drug-induced de-differentiation of melanoma cells, leading to adaptive resistance.
Specifically, they monitored the response of a melanoma cell line, COLO858, following exposure to the BRAF inhibitor Vemurafenib. 
Single-cell analysis and molecular profiling revealed an up-regulation of a de-differentiated NGFR$^{\text{High}}$ state in Vemurafenib-treated cells.
The authors also collected live-cell imaging bulk cell count data for the COLO858 cell line under various Vemurafenib concentration levels. 
This live-cell imaging data was provided in the follow-up work \cite{comandante2020phenotype}.
We now wish to show that our framework can be used to detect drug-induced plasticity in this dataset.
For consistency with previous definitions, we refer to the NGFR$^{\text{High}}$ and NGFR$^{\text{Low}}$ states of the COLO858 cell line as CSCs and CNSCs, respectively, for the remainder of this subsection.

\subsubsection{Data overview}

In {\cite{comandante2020phenotype}}, the COLO858 cell line was treated with six doses of Vemurafenib:
\[\mathcal{D}_v = \{0 (\text{DMSO}),0.032,0.1,0.32,1,3.2\} (\mu M),\]
over a period of $120$ hours.
The full dataset is shown in Figure \ref{fig:COLO858 data}. 

\begin{figure}
    \centering
    \includegraphics[width=0.8\linewidth]{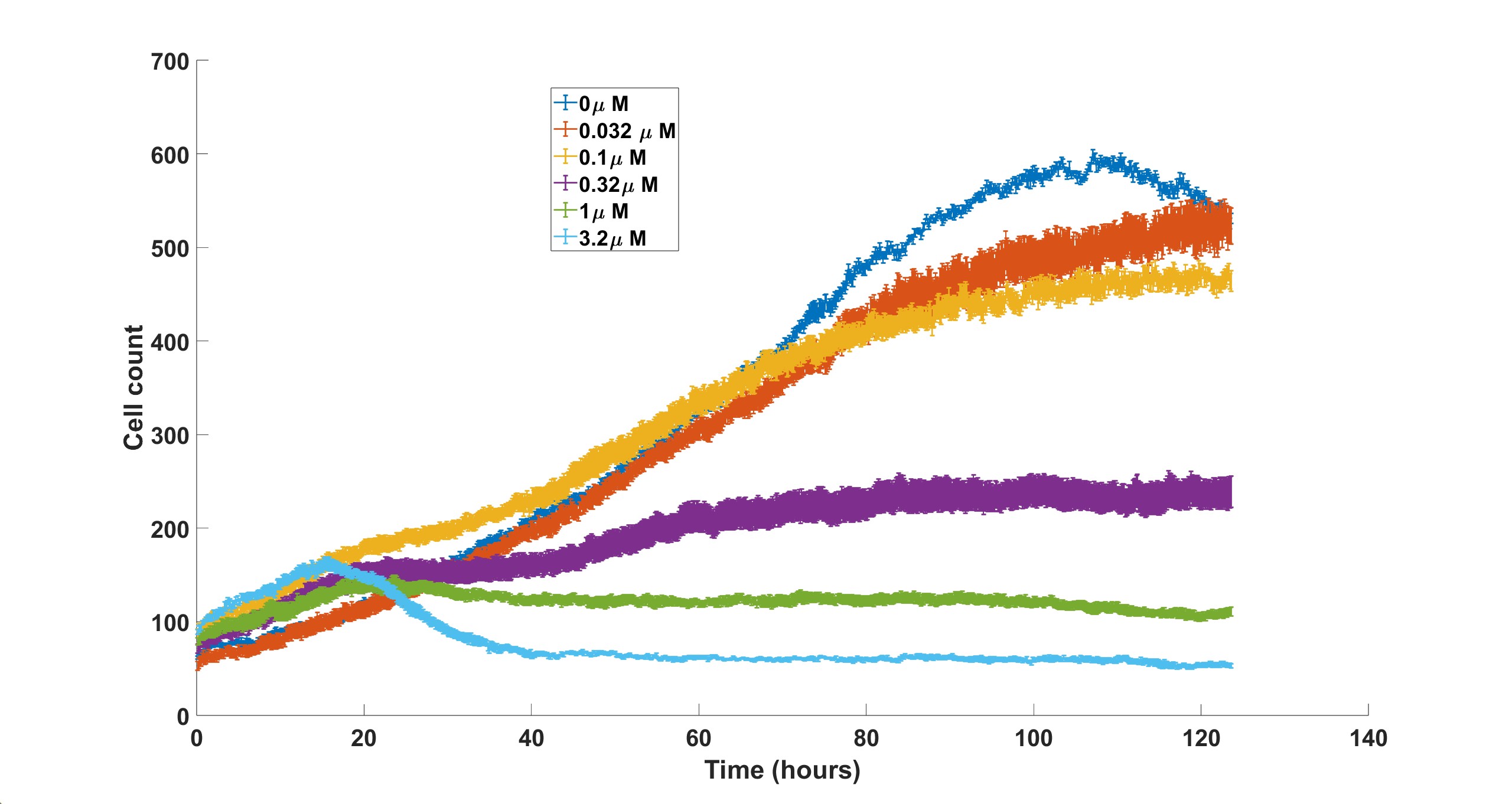}
    \caption{COLO858 bulk cell count data treated under six different concentration levels of Vemurafenib. Errorbars, color-coded for each drug concentration, represent the mean and scaled standard deviation from four replicates.}
    \label{fig:COLO858 data}
\end{figure}

Similar to the AGS--CPX-O dataset, there is a time-delayed drug effect at the beginning of these experiments. 
Fortunately, we have already proposed a time-delayed drug response model in Section {\ref{sec:time-delayed model}}.
It is worth reiterating that incorporating the time-delayed drug response model will complicate the computation of the \textit{mean matrix} and the \textit{covariance matrix}, so we decided to employ a simpler statistical model for analysis of this dataset as described in Section \ref{sec: Deterministic approximation and time-delayed drug-effect model}.

As is apparent from Figure \ref{fig:COLO858 data}, the current dataset exhibits a logistic growth dynamic under DMSO conditions, 
which is not consistent with the exponential growth assumption of our framework.
Therefore, we only use data from the first 60 hours of the experiments, which includes a total of 30 time points at 2-hour intervals, 
where
exponential growth is a reasonable assumption.

\subsubsection{Candidate model description and model selection results}

Similar to the AGS--CPX-O dataset, we propose using the AIC model selection criterion to determine whether the drug induces de-differentiation of the COLO858 cell line treated with Vemurafenib solely using live-cell imaging bulk cell count data.
Before defining various models of interest, we begin with two model assumptions derived from \cite{fallahi2017adaptive}:
\begin{itemize}
    \item Initially, the population consists entirely of melanoma CNSCs. ($p_s = 1, p_r = 0$).
    \item Vemurafenib may have cytotoxic effects on both CSCs and CNSCs; however, it promotes transitions only from CNSCs to CSCs. ($b_{r,\nu} = 1$).
\end{itemize}
These assumptions closely align with those made in Section {\ref{sec:CSC assumptions}}.
However, we would like to emphasize that we now allow the drug to have a cytotoxic effect on the CSC population, aligning with the observation in {\cite{fallahi2017adaptive}} that the COLO858 CSCs are less sensitive but still inhibited by the Vemurafenib.

We are fundamentally interested in determining whether the drug induces de-differentiation from CNSCs to CSCs, but since it is unknown whether there is natural plasticity between the two subpopulations or not, we will allow for both alternatives.
This leads to four candidate models.
The models $DIP_{nAsy}$ and $nDIP_{nAsy}$ assume no natural plasticity ($\nu_{rs}= \nu_{sr} = 0$), with the former model assuming the presence of drug-induced plasticity and the latter assuming no such drug effect.
The models $DIP_{Asy}$ and $nDIP_{Asy}$ assume the possibility of natural plasticity, where the former model again assumes the presence of drug-induced plasticity.
The free parameters for each model are provided in Table {\ref{tab:COLO858 model parameters}}.

\begin{table}
    \centering
    \begin{tabular}{|c|c|l|}
    \hline
        Model & $|\theta|$ & Free parameters  \\
        \hline
        $DIP_{nAsy}$ & 13 & $\alpha_r,\beta_r, b_{r,\beta},E_{r,\beta},\alpha_s,\beta_s,b_{s,\beta},E_{s,\beta},b_{s,\nu},E_{s,\nu},k,t_0,c$  \\
        \hline
        $nDIP_{nAsy}$ & 8 & $\alpha_s,\beta_s,b_{s,\beta},E_{s,\beta},m_{s,\beta},k,t_0,c$  \\
        \hline
        $DIP_{Asy}$ & 15 & $\alpha_r,\beta_r,\nu_{rs},b_{r,\beta},E_{r,\beta},\alpha_{s},\beta_{s},\nu_{sr},b_{s,\beta},E_{s,\beta},b_{s,\nu},E_{s,\nu},k,t_0,c$\\
        \hline
        $nDIP_{Asy}$ & 13 &     $\alpha_r,\beta_r,\nu_{rs},b_{r,\beta},E_{r,\beta},\alpha_{s},\beta_{s},\nu_{sr},b_{s,\beta},E_{s,\beta},k,t_0,c$\\
    \hline
    \end{tabular}
    \caption{Free parameters in models for COLO858--Vemurafenib experiment.}
    \label{tab:COLO858 model parameters}
\end{table}

\begin{table}
    \centering
    \begin{tabular}{|c|c|c|}
    \hline
         & Drug-induced plasticity & no Drug-induced plasticity \\
         \hline
        Without natural plasticity & \textbf{AIC($DIP_{nAsy}$): 6183} & AIC($nDIP_{nAsy}$): 6433\\
         \hline
        With natural plasticity & \textbf{AIC($DIP_{Asy}$): 6145}& AIC($nDIP_{Asy}$): 6370\\
        \hline
    \end{tabular}
    \caption{AIC value for 4 varying model assumptions in COLO858--Vemurafenib experiment.}
    \label{tab:COLO858 AIC value}
\end{table}

The AIC results are demonstrated in Table \ref{tab:COLO858 AIC value}.
See Appendix \ref{appx:In vitro details} for the details of the computation. 
We conclude that models assuming drug-induced plasticity are generally preferred, regardless of whether natural transitions between the two subpopulations are assumed.
These results show how our framework, using only data on the bulk cell population, is able to detect the presence of drug-induced resistance which has been confirmed using single-cell analysis in \cite{fallahi2017adaptive}.

In a recent paper {\cite{sontag2024understanding}}, Sontag et.al.~employed a deterministic model to analyze the same COLO858--Vemurafenib dataset.
Their estimation scheme assessed drug effects at different concentration levels with distinct parameters, estimating non-zero drug-induced plasticity.
In contrast, our statistical framework employed the Hill equation to quantify the dose-response relationship and concluded the presence of drug-induced plasticity across the entire dataset using the AIC model selection criteria.

\section{Conclusion}
\label{sec:Conclusion}

In this study, we introduced a novel statistical framework designed for analyzing HTS bulk data involving multiple subpopulations of cells.
Using an asymmetrical birth multi-type branching process,
we extended a model introduced in our previous work \cite{wu2024using} to 
accommodate transitions between the different subpopulations,
thereby enhancing its range of applications. 
One such application involves inferring drug-induced plasticity in a mixture population of cancer stem-like cells (CSCs) and cancer non-stem-like cells (CNSCs). 
The asymmetrical birth multi-type branching process naturally accounts for the differentiation of CSCs and the de-differentiation of CNSCs. 
Additionally, we incorporated the Hill equation to characterize the cytotoxic effects of anti-cancer drugs and drug-induced de-differentiation rates.
We tested our approach using both \textit{in silico} and \textit{in vitro} data.

In our \textit{in silico} experiments, we used stochastic simulation to generate datasets involving drug-treated mixtures of CSCs and CNSCs.
The datasets consisted of total cell counts collected at predetermined drug concentration levels and time points. 
Drawing inspiration from recent research on drug-induced plasticity \cite{padua2023high}, we formulated five assumptions to simulate the dynamics of CSCs and CNSCs, along with their respective drug effects. 
Under these assumptions, our newly proposed framework not only identifies the presence of drug-induced plasticity, but also accurately predicts several key features of the mixture dynamics.
Specifically, in an illustrative example, the framework accurately recovered the drug-affected \textit{stable proportions} between CSCs and CNSCs across various drug concentration levels. 
Furthermore, it determined the $GR_{50}$ values for both the drug's cytotoxic effect on CNSCs and its induced plasticity effect. 
Through further numerical experiments, we confirmed that the newly proposed framework consistently provides precise estimations for each model parameter, capturing the growth dynamics of individual subpopulations, transitions between the subpopulations, as well as drug-induced toxicity and plasticity.

We further investigated the robustness of our framework by relaxing assumptions regarding the dynamics of the CSCs and CNSCs and how the drug affects each supbopulation.
By allowing the CSC and CNSC mixture to start from an arbitrary initial proportion instead of the stable drug-free proportion, our framework accurately estimated not only the mixture's growth dynamics and drug effects but also the initial proportion. 
However, relaxing the assumption that the drug only affects CNSCs resulted in a more complex drug effect profile, which degraded the identifiability of each individual effect. 
In particular, the drug effects on CSCs were less identifiable, which may be due to our assumption that the drug has a more significant effect on CNSCs.
It should also be kept in mind that we only assume data on total cell counts, as opposed to cell counts for individual subpopulations, which is an inherent limitation for teasing apart complex evolutionary dynamics and drug effect profiles involving multiple parameters.
Using this kind of data, some simplifying assumptions must be made to ensure the identifiability of all model parameters.

Additionally, we explored the performance of our framework under the assumption of limited division of CNSCs, portraying the differentiation process as a gradual loss of cell proliferation potential. 
Our framework successfully estimated the CNSC division rate when the maximum division number $(G)$ of CNSCs equaled 0, 2 or 3.
Furthermore, it consistently provided reliable estimates of the half-maximum effects parameter ($E$) for each drug effect on the CNSCs. 
Even though parameter estimates generally degraded,
these results demonstrate that our newly proposed framework can offer valuable insights even under more complex and realistic scenarios.
To fully capture these dynamics, we propose expanding our model beyond two subpopulations, by modeling each generation of CNSCs as a distinct subpopulation.
Specifically, we can represent the symmetric division of $g$-th generation cells into two $(g+1)$-th generation cells as birth events where a cell of one type produces two cells of another type.
Currently, our asymmetric division framework does not account for this behavior, as it models subpopulation transitions through asymmetric divisions. 
However, modeling subpopulation transitions through symmetric division is straightforward using the multi-type branching process framework \cite{athreya2004branching}.
We plan to explore this in future work.

We validated our framework using data from the gastric carcinoma cell line (AGS) treated with ciclopirox olamine (CPX-O) {\cite{padua2023high}}, as well as BRAF-mutated melanoma cell line (COLO858) treated with Vemurafenib {\cite{comandante2020phenotype}}.
Based on single-cell tracking analysis, these two studies confirm the presence of drug-induced plasticity in both datasets.
Coupled with the AIC model selection scheme, our novel statistical framework is able to indicate the presence of drug-induced plasticity in both datasets using live-cell imaging bulk data alone.
Validation using data from \cite{padua2023high} posed challenges due to the limited \textit{in vitro} data available. 
This discrepancy is partly due to the experimental procedure described in \cite{padua2023high}, which involves using FASC to separate CSCs and CNSCs -- a step that our framework does not require.
Nevertheless, our framework indicates that CPX-O induces plasticity in the AGS gastric carcinoma cell line, aligning with the conclusions in \cite{padua2023high}.
In the COLO858--Vemurafenib validation experiment, an abundance of time-course data over a period of 120 hours is provided in {\cite{comandante2020phenotype}}.
However, we observed a logistic type of proliferation dynamics, potentially caused by carrying capacity, which is not assumed in our model.
Therefore, we chose to employ only the initial 60 hours data, which demonstrated a clear exponential type of proliferation.
Incorporating logistic proliferation dynamics is a potential future direction.

In general, our framework requires many parameters to model various drug effects, necessitating multiple data points at distinct concentration levels to adequately capture mixture dynamics under drug influence.
However, this requirement may be offset by the advantage of not needing a more sophisticated technique for subpopulation separation.
Moreover, the quality of the data is equally as important as the quantity of it.
As observed in our \textit{in silico} experiments, some true parameter sets became unidentifiable under the specific concentration levels used in the experiments. 
To enhance the framework's utility, one future direction involves developing techniques for experimental design, meaning optimal selection of drug concentrations and time points for accurate parameter inference.
Such analysis could markedly improve the framework's ability to extract insights from HTS data.

Another challenge identified in the \textit{in vitro} experiment was the time-inhomogeneous drug effect. 
While our drug-effect base model assumes time-homogeneous drug effects on the target tumor sample, data from \cite{padua2023high,comandante2020phenotype} show significant heterogeneity in drug effects over time. 
In Section \ref{sec:time-delayed model}, we addressed this by developing a two-parameter logistic function, resulting in a better fit to the \textit{in vitro} data.




There are several avenues for extension of the framework beyond those already mentioned.
First, while we have focused on inferring drug-induced plasticity in a CSC and CNSC mixture, the newly proposed framework is based on a versatile mathematical model (multi-type branching process) which can be adapted to a broader range of transition dynamics between two or more subpopulations. 
These subpopulations are defined by their heterogeneous responses to controllable interferences -- specifically, drugs in this case. 
Second, we have in this work focused on cellular responses to single-drug interventions.
One potential extension involves expanding our framework to encompass more complex multi-drug responses, potentially unveiling more intricate underlying systems.
Moreover, external factors such as nutrient levels, hypoxia, stromal content, and other microenvironmental factors could be integrated into our framework in future work.

\section*{Acknowledgements}
The authors would like to thank Diana P\'adua for helpful conversations regarding the analysis of the data in \cite{padua2023high}.
K.\ Leder, J.\ Foo and C.\ Wu were supported in part by  National Science Foundation, United States grant CMMI-2228034. 
J.\ Foo was supported in part by  National Science Foundation, United States grant DMS-2052465.

\section*{Author contributions}
All authors conceived of the study; C.W. and K.L. designed model; C.W. performed computational experiments; C.W. conducted data analysis; K.L. supervised the research; all authors wrote the paper.

\section*{Competing interests}
The authors declare no competing interests.

\section*{Code availability}
The code necessary to reproduce the experimental results presented in this paper are published in \url{https://github.com/chenyuwu233/Cancer-Stem-Cells-Drug-induced-Plasticity}.

\section{Appendix}
\label{sec:Appendix}

\subsection{Reducible \textit{infinitesimal generator matrix} $\mathbf{A}$ and long-run behavior of $\mathbf{B}(t,0)$}

\label{appx: Treatment under potential violation}

The results in Section \ref{sec: Long-run behavior} rely on the irreducibility assumption of \textit{infinitesimal generator matrix} $\mathbf{A}$. However, Assumption \ref{assump:Cell plasticity} in Section \ref{sec:CSC assumptions} violates this irreducibility assumption under no-drug conditions. Consequently, we explore the details of these assumptions and discuss the situation when those assumptions are violated in this section. 

The main issue arising from reducibility is that the process $\mathbf{B}(t,0)$ may not converge almost surely to a \textit{stable proportion} $\bm{\pi}$ in the long-run.
To address this, we use the results in \cite{nicholson2023sequential}. 
Given that we assume $\kappa_s > \kappa_d > 0$ (Assumption \ref{assump:Cell dynamic}) and CSCs do not go extinct (Assumption \ref{assump:no extinct CSCs}), we can directly apply Proposition 1 and Corollary 1 in \cite{nicholson2023sequential}. 
Specifically, Proposition 1 in \cite{nicholson2023sequential} states that there exist random variables $V_{s}$ and $V_d$ such that
\[\lim_{t \rightarrow \infty} e^{-\kappa_s t} B_{s}(t) = V_s \text{ and } \lim_{t\rightarrow \infty} e^{-\kappa_s t} B_d(t) = V_d\]
almost surely.
From Corollary 1 in \cite{nicholson2023sequential}, we further have 
\[V_{d} = \frac{\nu_{sd} V_s}{\kappa_s - \kappa_d}\]
almost surely. 
From this ratio, we can further derive
\[[V_s,V_d] = C \left[\frac{\kappa_s - \kappa_d}{\kappa_s - \kappa_d+\nu_{sd}}, \frac{\nu_{sd}}{\kappa_s - \kappa_d+\nu_{sd}}\right] = C \bm{\pi},\]
where $C$ is a random variable and the \textit{stable proportion} $\bm{\pi}$ is still the left eigenvector of $\mathbf{A}$ corresponding to the largest eigenvalue.
Consequently, the \textit{stable proportion} exists. By following the same procedure---obtaining the left eigenvector of $\mathbf{A}$ corresponding to the largest eigenvalue---we can determine the \textit{stable proportion} when $\mathbf{A}$ becomes reducible under the no-drug condition.

\subsection{Long-run behavior of finite moment of the multi-type branching process} 
\label{appx: finite moment}

In \cite{athreya2004branching}, the authors reviewed the \textit{long-run growth rate} of the first and the second moments of the multi-type branching process. Specifically, the \textit{mean matrix} scaled by $e^{-\lambda_1 t}$ would converge to a matrix independent of time. In other words, the long-run expected cell count will increase at rate $e^{\lambda_1 t}$. Similarly, the second moment has a growth rate $e^{2\lambda_1 t}$. To prove Proposition \ref{prop: Lindeberg CLT}, we extend these results by considering any finite moment of this process, i.e. we want to show that for any positive finite integer $z$ and combination $\mathbf{c}_z: [c_1(z),\cdots, c_K(z)] \geq 0$ such that $ \sum_{i = 1}^{K} c_{i}(z) = z$,
\[\limsup_{t\rightarrow \infty} \expect*{B^{(i)}_1(t)^{c_{1(z)}} B_2^{(i)}(t)^{c_{2(z)}}\cdots B_K^{(i)}(t)^{c_{K(z)}}} e^{-\lambda_1 z t} < \infty. \]
Consequently, we show the following Lemma \ref{lemma: Finite moment}.



\begin{lemma}
    \label{lemma: Finite moment}
    For any integer $z >0$ and combination $\mathbf{c}_z$, we have
    \[\expect*{B^{(i)}_1(t)^{c_{1(z)}} B_2^{(i)}(t)^{c_{2(z)}}\cdots B_K^{(i)}(t)^{c_{K(z)}}} = O(e^{\lambda_1 zt}).\]
\end{lemma}

\begin{proof}

The probability-generating function for the multi-type branching process for $\mathbf{B}^{(i)}(t)$ is denoted as
\[F^{(i)}(\mathbf{s},t) := \expect*{\mathbf{s}^{\mathbf{B}^{(i)}(t)}} = \expect*{\prod_{k = 1}^{K}s_k^{\mathbf{B}_k^{(i)}(t)}},\]
where $\mathbf{s} = [s_1,\cdots, s_K]$. With these notations, we can write the Kolmogorov forward equation \cite{athreya2004branching} for $\mathbf{B}^{(i)}(t)$ as
\begin{equation}
\label{eq: KFE}
\frac{\partial}{\partial t} F^{(i)}(\mathbf{s},t) = \sum_{k = 1}^{K} u^{(k)}(\mathbf{s}) \frac{\partial}{\partial s_k} F^{(i)}(\mathbf{s},t),
\end{equation}
where the function $u^{(i)}(s)$ for type $i$ is defined as
\[u^{(i)}(\mathbf{s}):= \alpha_i s_i^2 + \beta_i + \sum_{j \neq i} \nu_{ij} s_is_j - \left(\alpha_i + \beta_i + \sum_{j \neq i} \nu_{ij}\right) s_i , \quad 0\leq \mathbf{s} \leq 1 \]

We denote a vector $\mathbf{c}_{z} = [c_{1(z)},c_{2(z)},\cdots,c_{K(z)}]$ as the combination index vector. Then we use $\partial \mathbf{s}^{\mathbf{c}_z}$ to represent $\partial s_1^{c_{1(z)}} \partial s_2^{c_{2(z)}}\cdots \partial s_K^{c_{K(z)}}$. 

Now we can define
\begin{equation}
\label{eq: Finite moment surrogate}
\ell_z(t) = \max_{\mathbf{c}_z}\frac{\partial^z}{\partial \mathbf{s}^{\mathbf{c}_z}} F^{(i)}(\mathbf{s},t)\Big\rvert_{\mathbf{s} =  \mathbf{1}},
\end{equation}
and we denote $\mathbf{c}_{z}^*$ as the maximum combination index of $z$-th moment. We can see from the equation \eqref{eq: Finite moment surrogate} that for each $z\geq 1$ there exists $C_z' >0$ such that for any combination $\mathbf{c}_z$ 
\begin{align*}
    \ell_{z}(t) &\geq \frac{\partial^z}{\partial \mathbf{s}^{\mathbf{c}_z}} F^{(i)}(\mathbf{s},t)\Big\rvert_{\mathbf{s} = 1} \\
    &\geq C_z' \sum_{n_1 = c_{1(z)},\cdots,n_K = c_{K(z)}}^{\infty}  n_1^{c_{1(z)}}\cdots n_K^{c_{K(z)}} P(B^{(i)}_1(t) = n_1,\cdots,B^{(i)}_K(t) = n_K) \\
    &= C_z' \expect*{\prod_{k = 1}^{K}B^{(i)}_k(t)^{c_{k(z)}}} - C_z' \expect*{\prod_{k = 1}^{K}B^{(i)}_k(t)^{c_{k(z)}}; B^{(i)}_k(t) < c_{k(z)} }.
\end{align*}
As the multi-type branching process goes to infinity or extinct with probability 1 as time goes to infinity, we have, in the supercritical case, 
\[\lim_{t\rightarrow \infty} \frac{\expect*{\prod_{k = 1}^{K}B^{(i)}_k(t)^{c_{k(z)}^{*}}; B^{(i)}_k(t) < c_{k(z)}^* }}{\expect*{\prod_{k = 1}^{K}B^{(i)}_k(t)^{c_{k(z)}^{*}}}} = 0\]
so that there exists $C_z$ such that $\ell_z(t)\geq C_z \expect*{\prod_{k = 1}^{K}B^{(i)}_k(t)^{c_{k(z)}^{*}}}$.

Thus, it suffices to establish the upper bound on $\ell_z(t)$ to prove the Lemma. We show the upper bound of $\ell_{z}(t)$ by induction on $z$. The base case was given by the fact that $e^{\lambda_1 t}$ is the largest eigenvalue of the \textit{mean matrix} $\mathbf{M}(t)$, i.e.
\[\expect*{B^{(i)}_k(t)} = [\mathbf{e}_i \mathbf{M}(t)]_{k} \leq e^{\lambda_1 t},\]
where $[\cdot]_k$ is the $k$-th element of the vector, for any $k \in \{1,\cdots,K\}$. Then we assume that for $z > 1$ and all $j \leq z-1, \ell_{j}(t) = O(e^{\lambda_1 j t})$. Then from the equation \eqref{eq: KFE} we have that
\begin{align*}
    \frac{d}{dt}\ell_{z} (t) &= \frac{\partial^z}{\partial \mathbf{s}^{\mathbf{c}_z^*}}\sum_{k = 1}^{K}  u^{(k)}(\mathbf{s}) \frac{\partial}{\partial s_{k}} F(\mathbf{s},t)\Big\rvert_{\mathbf{s} = \mathbf{1}}\\
    &= \sum_{j = 0}^{z} \sum_{\mathbf{c}_j \leq \mathbf{c}_z^*} {c_{1(z)}^{*} \choose c_{1(j)}}\cdots {c_{K(z)}^{*} \choose c_{K(j)}} \sum_{k = 1}^{K} \frac{\partial^j}{\partial \mathbf{s}^{\mathbf{c}_j}}  u^{(k)}(\mathbf{s}) \frac{\partial^{z-j+1}}{\partial \mathbf{s}^{\mathbf{c}_z^* - \mathbf{c}_{j}} \partial s_k} F(\mathbf{s},t) \Big\rvert_{\mathbf{s} = \mathbf{1}}\\
    &= \sum_{j = 1}^{z} \sum_{\mathbf{c}_j  \leq \mathbf{c}_z^*} {c_{1(z)}^{*} \choose c_{1(j)}}\cdots {c_{K(z)}^{*} \choose c_{K(j)}} \sum_{k = 1}^{K} \frac{\partial^j}{\partial \mathbf{s}^{\mathbf{c}_j}}  u^{(k)}(\mathbf{s}) \frac{\partial^{z-j+1}}{\partial \mathbf{s}^{\mathbf{c}_z^* - \mathbf{c}_{j}} \partial s_k} F(\mathbf{s},t) \Big\rvert_{\mathbf{s} = \mathbf{1}}\\
    &\leq \sum_{j = 1}^{z} \sum_{\mathbf{c}_j \leq \mathbf{c}_z^*} {c_{1(z)}^{*} \choose c_{1(j)}}\cdots {c_{K(z)}^{*} \choose c_{K(j)}} \sum_{k = 1}^{K} \frac{\partial^j}{\partial \mathbf{s}^{\mathbf{c}_j}}  u^{(k)}(\mathbf{1}) \ell_{z-j+1}(t)\\
    &\leq \lambda_1 z \ell_{z} (t) + \sum_{j = 2}^{z} \sum_{\mathbf{c}_j \leq \mathbf{c}_z^*} {c_{1(z)}^{*} \choose c_{1(j)}}\cdots {c_{K(z)}^{*} \choose c_{K(j)}} \sum_{k = 1}^{K} \frac{\partial^j}{\partial \mathbf{s}^{\mathbf{c}_j}}  u^{(k)}(\mathbf{1}) \ell_{z-j+1}(t)
\end{align*}
The third equality is due to the fact that $u^{(k)}(\mathbf{1}) = 0$ for all $k$, the first inequality follows from the definition of $\ell_z(t)$, and the last inequality is due to the following inequality
\[\sum_{k = 1}^K \frac{\partial}{\partial s_i} u^{(k)}(\mathbf{1})= \alpha_i - \beta_i + \sum_{k \neq i} \nu_{ki} = \mathbf{1} \mathbf{A} \mathbf{e}_i^{\intercal} \leq \lambda_1,\]
for any $i \in \{1,\cdots,K\}$. Since $\ell_{z - j+1}(t) = O(e^{\lambda_1 (z-j+1) t})$, for $z \geq j \geq 2$, there exists constants $\alpha_{z,\mathbf{c}_j}$ such that 
\begin{align*}
    e^{-\lambda_1 z t} \ell_z(t) &\leq \sum_{j = 2}^{z}\sum_{\mathbf{c}_j\leq \mathbf{c}_z^*}{c_{1(z)}^{*} \choose c_{1(j)}}\cdots {c_{K(z)}^{*} \choose c_{K(j)}}  \sum_{k = 1}^{K} \alpha_{z,\mathbf{c}_j} \int_{0}^{t} e^{-\lambda_1 \tau(j-1)} d\tau  \\
    &= \sum_{j = 2}^{z} \sum_{\mathbf{c}_j\leq \mathbf{c}_z^*} {c_{1(z)}^{*} \choose c_{1(j)}}\cdots {c_{K(z)}^{*} \choose c_{K(j)}}  \sum_{k = 1}^{K} \frac{\alpha_{z,\mathbf{c}_j}}{\lambda_1 (j-1)} \left(1 - e^{-\lambda_1 t (j-1)} \right) .
\end{align*}
The induction hypothesis is thus established, completing the proof.

\end{proof}

\subsection{Derivation of the covariance for asymmetric birth model:}
\label{appx: Covariance}

In this sub-section, we would like to derive the covariance matrix 
 \[\mathbf{\Xi}^{(i)}(t) = \expect*{ \left( \mathbf{B}^{(i)}(t) - \mathbf{m}^{(i)}(t)\right) ^{\intercal}  \left( \mathbf{B}^{(i)}(t) - \mathbf{m}^{(i)}(t)\right)}.\]
Rather than drive this directly, we denote a second factorial moment $\mathbf{D}_{jk}^{(i)}(t) = \expect*{\mathbf{B}^{(i)}_j(t)(\mathbf{B}^{(i)}_k(t) - \delta_{jk})}$, where $\delta_{jk}$ is the Kronecker delta. Note that
\begin{equation}
\label{eq: second fac to covariance}
\mathbf{D}^{(i)}(t) = \mathbf{\Xi}^{(i)}(t) + \mathbf{m}^{(i)}(t) ^{\intercal} \mathbf{m}^{(i)}(t)  - \diag(\mathbf{m}^{(i)}(t))
\end{equation}

From the proof of Lemma \ref{lemma: Finite moment}, we have the probability generating function for $B^{(i)}(t)$ defined as
\[F^{(i)}(\mathbf{s},t) := \expect*{\mathbf{s}^{B^{(i)}(t)}} = \expect*{\prod_{k = 1}^{K} s_k^{B^{(i)}_k(t)} }, \]
where $\mathbf{s} = [s_1,\cdots,s_K]$. Then from the Kolmogorov Forward Equation \eqref{eq: KFE}, we have the $u^{(i)}(\mathbf{s})$ for asymmetric birth model defined as:
\begin{equation}
\label{eq: function u}
u^{(i)}(\mathbf{s}) := \alpha_i s_i^2 + \beta_i + \sum_{j\neq i} \nu_{ij} s_i s_j - \left(\alpha_i + \beta_i + \sum_{j \neq i} \nu_{ij} \right)s_i, \quad 0\leq \mathbf{s}\leq 1.
\end{equation}
Then we have 
\begin{equation}
\label{eq: KFE second}
\begin{split}
    \frac{\partial}{\partial t} \left(\frac{\partial^2}{\partial s_j\partial s_k} F^{(i)}(\mathbf{s},t)  \right)  &= \sum_{\ell = 1}^{K} \left[ \frac{\partial^2}{\partial s_j \partial s_k} u^{(\ell)}(\mathbf{s}) \frac{\partial}{\partial s_{\ell}} F^{(i)}(\mathbf{s},t) + \frac{\partial}{\partial s_j} u^{(\ell)}(\mathbf{s}) \frac{\partial^2}{\partial s_\ell \partial s_k} F^{(i)}(\mathbf{s},t) \right.\\
    & \left.+ \frac{\partial }{\partial s_k} u^{(\ell)}(\mathbf{s}) \frac{\partial^2}{\partial s_{\ell} \partial s_j} F^{(i)}(\mathbf{s},t) + u^{(\ell)}(\mathbf{s}) \frac{\partial^3}{\partial s_\ell \partial s_j \partial s_k} F^{(i)}(\mathbf{s},t)\right].
\end{split}
\end{equation}
We can see from the equation \eqref{eq: function u} that for any $\ell \in \{1,\cdots,K\}$
\begin{align*}
    u^{(\ell)}(\mathbf{1}) &= 0\\
    \frac{\partial }{\partial s_j} u^{(\ell)}(\mathbf{1}) &= \begin{cases}
        \alpha_j - \beta_j & j = \ell\\
        \nu_{\ell i} & j \neq \ell
    \end{cases}\\
    \frac{\partial^2}{\partial s_j \partial s_k} u^{(\ell)}(\mathbf{1}) &= \begin{cases}
        \nu_{jk} & \ell = j \neq k\\
        2 \alpha_j & \ell = j = k.
    \end{cases}
\end{align*}
Letting $\mathbf{s} = \mathbf{1}$ and re-arranging equation \eqref{eq: KFE second}, we can see that
\[\frac{d}{dt} \mathbf{D}_{jk}^{(i)}(t) = \sum_{\ell = 1}^{K} \mathbf{A}_{\ell j} \mathbf{D}_{k\ell}^{(i)}(t) + \sum_{\ell = 1}^{K} \mathbf{A}_{\ell k} \mathbf{D}_{j\ell}^{(i)}(t) + (1-\delta_{jk})[\nu_{kj} \mathbf{M}_{ik}(t) + \nu_{jk}\mathbf{M}_{ij}(t)] + 2\delta_{jk} \alpha_j \mathbf{M}_{ij}(t),\]
where $\mathbf{A}$ is the \textit{infinitesimal generator matrix}, and we know that $\mathbf{D}^{(i)}(t)$ is a symmetric matrix. Thus, we derive the following matrix differential equation
\begin{equation}
\label{eq: KFE second matrix}
\frac{d}{dt} \mathbf{D}^{(i)}(t) = \mathbf{A}^{\intercal}\mathbf{D}^{(i)}(t) + \mathbf{D}^{(i)}(t)\mathbf{A} + \mathbf{C}^{(i)}(t),
\end{equation}
where $\mathbf{C}^{(i)}_{jk}(t) = \begin{cases}
    \nu_{kj} \mathbf{M}_{ik}(t) + \nu_{jk}\mathbf{M}_{ij}(t) & j \neq k\\
    2\alpha_j \mathbf{M}_{ij}(t) & j = k
\end{cases}$. Observe that equation \eqref{eq: KFE second matrix} is a Lyapunov matrix differential equation. With the initial condition $\mathbf{D}^{(i)}(0) = 0$, we have the solution 
\begin{align*}
    \mathbf{D}^{(i)}(t) &= \exp(t\mathbf{A}^{\intercal}) \left(\int_{0}^{t} \exp(-\tau \mathbf{A}^{\intercal}) \mathbf{C}^{(i)}(\tau)  \exp(-\tau \mathbf{A}) d \tau \right) \exp(t \mathbf{A})\\
    &= \int_{0}^{t} \left(\mathbf{M}(t - \tau) \right)^{\intercal} \mathbf{C}^{(i)}(\tau) \left(\mathbf{M}(t-\tau) \right) d\tau.
\end{align*}
Then we can derive the covariance $\mathbf{\Xi}^{(i)}(t)$ from $\mathbf{D}^{(i)}(t)$ from equation \eqref{eq: second fac to covariance}.


\subsection{Extension of Proposition \ref{prop: finite CLT} }

\label{appx: Lindeberg CLT}

In Proposition \ref{prop: finite CLT}, we assume that time points $\mathcal{T} = \{t_1,\cdots,t_{N_T}\}$ are independent of the initial cell count $n$. This assumption becomes invalid if we assign specific numerical values to $n$ and $t_1,\cdots t_{N_T}$. 
Therefore, in this subsection, we aim to relax this assumption to make our results applicable in a more realistic scenario, where there is a dependency between the set of time points and the initial cell count.

To that end, we denote
\[\mathbf{Y}(\mathcal{T}_n):= [Y(t_{1,n}),Y(t_{2,n}), \cdots,  Y(t_{N_T,n})],\]
where $\mathcal{T}_n$ depends on $n$ and $t_{i,n}$ belongs to one of the following sets of sequences:
\begin{align*}
    \mathcal{F} & := \{ (t_{n})_{n\in \mathbb{N}} ; \limsup_{n\rightarrow \infty} t_{n} <\infty\}\\
    \mathcal{I} & := \{ (t_{n})_{n\in \mathbb{N}}; \lim_{n\rightarrow \infty} t_{n} = \infty\}.
\end{align*}
The sets $\mathcal{F}$ and $\mathcal{I}$ correspond to the finite time points sequence, which remain finite as $n\rightarrow \infty$, and the infinite time points sequence, which diverge to infinite as $n\rightarrow \infty$, respectively. 




Since we introduce the dependency on $n$, the covariance between $Z(t_{i,n})$ and $Z(t_{j,n})$ no longer remains constant as $n$ approaches infinity. Therefore, directly applying the multivariate CLT is not valid in this case; instead, we should utilize the more general Lindeberg-Feller multivariate CLT.

To apply the Lindeberg-Feller multivariate CLT, we denote the normalized random vector
\[\bar{\mathbf{Y}}(\mathcal{T}_n) = \sum_{i = 1}^{K} \sum_{j = 1}^{n_i} \bar{\mathbf{y}}^{(i)}_{j}(\mathcal{T}_n) = \sum_{i = 1}^{K} \sum_{j = 1}^{n_i} [\exp(-\lambda_1 t_{1,n})\langle \mathbf{B}^{(i)}_j(t_{1,n}),\mathbf{1}\rangle,\cdots, \exp(-\lambda_1 t_{N_T,n})\langle \mathbf{B}^{(i)}_j(t_{N_T,n}),\mathbf{1}\rangle].\]
We next denote the 
\[\bar{\bm{\mu}}^{(i)}(\mathcal{T}_n) = [\exp(-\lambda_1 t_{1,n}) \langle\mathbf{m}^{(i)}(t_{1,n}), \mathbf{1}\rangle,\cdots,\exp(-\lambda_1 t_{N_T,n}) \langle \mathbf{m}^{(i)}(t_{N_T,n}), \indi \rangle]\]
as the expectation of $\bar{\mathbf{y}}^{(i)}_{j}(\mathcal{T}_n)$. Note that all elements in $\bar{\bm{\mu}}^{(i)}(\mathcal{T}_n)$ are bounded above by 1 according to the spectral property of the \textit{mean matrix}. 

Same as Proposition \ref{prop: finite CLT}, we define a centered and normalized process for a given $i \in \mathcal{K}$
\[\bar{\mathbf{W}}_{n_i}^{(i)}(\mathcal{T}_n) = \frac{1}{\sqrt{n_i}} \sum_{j=1}^{n_i}\left(\bar{\mathbf{y}}_{j}^{(i)}(\mathcal{T}_n) - \bar{\bm{\mu}}^{(i)}(\mathcal{T}_n) \right)\]
and claim Proposition \ref{prop: Lindeberg CLT} from Lindeberg-Feller multivariate CLT







\begin{proposition}
\label{prop: Lindeberg CLT}
    as $n \rightarrow \infty$
    \[\bar{\mathbf{W}}_{n_i}^{(i)}(\mathcal{T}_n) \Rightarrow N(0, \bar{\mathbf{V}}^{(i)}(\mathcal{T}_n))\]
    where for $1\leq a\leq b\leq N_T$
    \[\bar{\mathbf{V}}^{(i)}_{a,b}(\mathcal{T}_n) = \begin{cases} \mathbf{1} \exp(-2\lambda_1 t_{a,n})\mathbf{\Xi}^{(i)}(t_{a,n}) \exp(-\lambda_1(t_{b,n} - t_{a,n}))\mathbf{M}(t_{b,n} - t_{a,n}) \mathbf{1}^{\intercal} & t_{a,n}, t_{b,n} \in \mathcal{F}\\
\mathbf{1} \exp(-2\lambda_1 t_{a,n})\mathbf{\Xi}^{(i)}(t_{a,n}) \mathbf{P} \mathbf{1}^{\intercal} &  t_{b,n} - t_{a,n} \in \mathcal{I}, t_{a,n} \in \mathcal{F}\\
\mathbf{1} \mathbf{Q}^{(i)} \exp(-\lambda_1(t_{b,n} - t_{a,n}))\mathbf{M}(t_{b,n} - t_{a,n}) \mathbf{1}^{\intercal} &  t_{b,n} - t_{a,n}\in \mathcal{F}, t_{a,n} \in \mathcal{I}\\
\mathbf{1} \mathbf{Q}^{(i)}\mathbf{P} \mathbf{1}^{\intercal} & t_{b,n} - t_{a,n} \in \mathcal{I}, t_{a,n} \in \mathcal{I}.
            \end{cases}\]
    The $\mathbf{Q}^{(i)}$ and $\mathbf{P}$ are the limit of the $\exp(-2\lambda_1 t) \mathbf{\Xi}^{(i)}(t)$ and $\exp(-\lambda_1 t)\mathbf{M}(t)$ respectively as $t\rightarrow \infty$.
\end{proposition}

\begin{proof}
    For simplicity, we denote $\mathbf{z}_{j}^{(i)}(\mathcal{T}_n) = \bar{\mathbf{y}}_{j}^{(i)}(\mathcal{T}_n) - \bar{\bm{\mu}}^{(i)}(\mathcal{T}_n)$. It is obvious that $\expect*{\mathbf{z}_j^{(i)}(\mathcal{T}_n)} = \mathbf{0}$ for all $n$. Then we need to check the behavior of its covariance matrix. Without loss of generality, we assume $1 \leq a\leq b\leq N_T$. For $a,b$-th element of the second moment of $\mathbf{z}_j^{(i)}(\mathcal{T}_n)$:
\begin{align*}
    [\mathbb{V} \mathbf{z}_j^{(i)}(\mathcal{T}_n)]_{a,b} &= 
    \exp(-\lambda_1 (t_{a,n} + t_{b,n}))\Cov\left(\langle \mathbf{B}^{(i)}(t_{a,n}), \mathbf{1}\rangle,\langle \mathbf{B}^{(i)}(t_{b,n}), \mathbf{1}\rangle \right)\\
   \Cov\left(\langle \mathbf{B}^{(i)}(t_{a,n}), \mathbf{1}\rangle,\langle \mathbf{B}^{(i)}(t_{b,n}), \mathbf{1}\rangle \right) &= \mathbf{1} \left\{\expect*{\mathbf{B}^{(i)}(t_{a,n})^{\intercal} \mathbf{B}^{(i)}(t_{b,n})} - \expect*{\mathbf{B}^{(i)}(t_{a,n})}^{\intercal} \expect*{\mathbf{B}^{(i)}(t_{b,n})}\right\} \mathbf{1}^{\intercal}  \\
    &= \mathbf{1} \left\{\expect*{\mathbf{B}^{(i)}(t_{a,n})^{\intercal} \expect*{\mathbf{B}^{(i)}(t_{b,n}) | \mathbf{B}^{(i)}(t_{a,n})}}\right. \\
    &\quad \left. - \expect*{\mathbf{B}^{(i)}(t_{a,n})}^{\intercal}\expect*{\expect*{\mathbf{B}^{(i)}(t_{b,n})|\mathbf{B}^{(i)}(t_{a,n})}}\right\}\mathbf{1}^{\intercal}\\
    &= \mathbf{1} \left\{\expect*{\mathbf{B}^{(i)}(t_{a,n})^{\intercal} \mathbf{B}^{(i)}(t_{a,n}) \mathbf{M}(t_{b,n} - t_{a,n})}\right. \\
    &\quad \left. - \expect*{\mathbf{B}^{(i)}(t_{a,n})}^{\intercal}\expect*{\mathbf{B}^{(i)}(t_{a,n}) \mathbf{M}(t_{b,n} - t_{a,n})}\right\}\mathbf{1}^{\intercal}\\
    &= \mathbf{1} \left\{ \expect*{\mathbf{B}^{(i)}(t_{a,n})^{\intercal} \mathbf{B}^{(i)}(t_{a,n})} - \expect*{\mathbf{B}^{(i)}(t_{a,n})}^{\intercal} \expect*{\mathbf{B}^{(i)}(t_{a,n})} \right\} \mathbf{M}(t_{b,n} - t_{a,n})\mathbf{1}^{\intercal}\\
    &= \mathbf{1} \mathbf{\Xi}^{(i)}(t_{a,n}) \mathbf{M}(t_{b,n} - t_{a,n})\mathbf{1}^{\intercal}\\
    [\mathbb{V} \mathbf{z}_j^{(i)}(\mathcal{T}_n)]_{a,b} &=  \mathbf{1} \exp(-2\lambda_1 t_{a,n})\mathbf{\Xi}^{(i)}(t_{a,n}) \exp(-\lambda_1 (t_{b,n} - t_{a,n}))\mathbf{M}(t_{b,n} - t_{a,n})\mathbf{1}^{\intercal}
\end{align*}
We know from Athreya and Ney that:
\begin{align*}
    \lim_{t\rightarrow \infty} \mathbf{M}(t) \exp(-\lambda_1 t) = \mathbf{P}
\end{align*}
where $\mathbf{P}$ is the outer product of the left and right eigenvector of $\mathbf{M}(t)$ corresponding to eigenvalue $\exp(\lambda_1 t)$, and there exists a finite matrix $\mathbf{Q}^{(i)}$ such that
\begin{align*}
    \lim_{t\rightarrow \infty} \mathbf{\Xi}^{(i)}(t) \exp(-2\lambda_1 t) = \mathbf{Q}^{(i)}.
\end{align*}
Therefore, we know that all the elements in $\mathbb{V}\mathbf{z}_j^{(i)}(\mathcal{T}_n)$ are finite. In particular, we claim the following condition for Lindeberg-Feller multivariate CLT holds:
\begin{align*}
    \lim_{n_i\rightarrow \infty} \frac{1}{n_i} \sum_{j = 1}^{n_i} \mathbb{V} \mathbf{z}^{(i)}_{j}(\mathcal{T}_n) &=\lim_{n\rightarrow \infty} \mathbb{V} \mathbf{z}^{(i)}_{j}(\mathcal{T}_n) = \bar{\mathbf{V}}^{(i)}(\mathcal{T}_n),
\end{align*}
where
\[\bar{\mathbf{V}}^{(i)}_{a,b}(\mathcal{T}_n) = \begin{cases} \mathbf{1} \exp(-2\lambda_1 t_{a,n})\mathbf{\Xi}^{(i)}(t_{a,n}) \exp(-\lambda_1(t_{b,n} - t_{a,n}))\mathbf{M}(t_{b,n} - t_{a,n}) \mathbf{1}^{\intercal} & t_{a,n}, t_{b,n} \in \mathcal{F}\\
\mathbf{1} \exp(-2\lambda_1 t_{a,n})\mathbf{\Xi}^{(i)}(t_{a,n}) \mathbf{P} \mathbf{1}^{\intercal} &  t_{b,n} - t_{a,n} \in \mathcal{I}_{\tau}\\
\mathbf{1} \mathbf{Q} \exp(-\lambda_1(t_{b,n} - t_{a,n}))\mathbf{M}(t_{b,n} - t_{a,n}) \mathbf{1}^{\intercal} &  t_{a,n} \in \mathcal{I}\\
\mathbf{1} \mathbf{Q}\mathbf{P} \mathbf{1}^{\intercal} & t_{b,n} - t_{a,n} \in \mathcal{I}_{\tau}, t_{a,n} \in \mathcal{I}.
\end{cases}\]
The first equality is due to the following two observations: 1. As $p_i$ independent of $n$, we know that $n_i\rightarrow \infty$ is equivalent to $n\rightarrow \infty$. 2. For any given $n$, the random vector $\mathbf{z}_j^{(i)}(\mathcal{T}_n)$ is i.i.d. copy of $\mathbf{z}^{(i)}(\mathcal{T}_n)$.

Then we claim the Lindeberg condition: 
\begin{equation}
\label{eq: Lindeberg}
    \lim_{n_i\rightarrow \infty}\frac{1}{n_i} \sum_{j = 1}^{n_i} \expect*{\|\mathbf{z}^{(i)}_{j}(\mathcal{T}_n)\|^2 \indi(\|\mathbf{z}^{(i)}_{j}(\mathcal{T}_n)\|\geq \epsilon \sqrt{n})} =  0
\end{equation}
for a given $\epsilon > 0 $, where $\indi(\cdot)$ is the indicator function, is satisfied. Due to the aforementioned two observations, it is equivalent to show that
\begin{equation}
\label{eq: Lindeberg reduced}
    \lim_{n\rightarrow \infty}\expect*{\|\mathbf{z}^{(i)}(\mathcal{T}_n)\|^2 \indi(\|\mathbf{z}^{(i)}(\mathcal{T}_n)\|\geq \epsilon \sqrt{n})} = 0.
\end{equation}
We applied the Cauchy-Schwarz inequality to \eqref{eq: Lindeberg reduced}, which simplifies the task to showing that
\begin{equation}
\label{eq: Lindeberg Cauchy-schwarz}
    \lim_{n\rightarrow \infty}\expect*{\|\mathbf{z}^{(i)}(\mathcal{T}_n)\|^4}^{1/2} \expect*{\indi(\|\mathbf{z}^{(i)}(\mathcal{T}_n)\| \geq \epsilon \sqrt{n})}^{1/2} = 0.
\end{equation}
Then we argue that
\begin{align*}
    \lim_{n\rightarrow \infty}\expect*{\|\mathbf{z}^{(i)}(\mathcal{T}_n)\|^4}  &= \lim_{n \rightarrow \infty} \expect*{\left[\sum_{\tau = 1}^{N_T} \exp(-2\lambda_1 t_{\tau,n})\left(\langle \mathbf{B}^{(i)}(t_{\tau,n}),\mathbf{1}\rangle - \langle\mathbf{m}^{(i)}(t_{\tau,n}),\mathbf{1}\rangle\right)^2\right]^2}\\
    &\leq \lim_{n \rightarrow \infty} \expect*{\left[\sum_{\tau = 1}^{N_T} \exp(-2\lambda_1 t_{\tau,n}) (\mathbf{1} \mathbf{B}^{(i)}(t_{\tau,n})^{\intercal} \mathbf{B}^{(i)}(t_{\tau,n}) \mathbf{1}^{\intercal} + \mathbf{1} \mathbf{m}^{(i)}(t_{\tau,n})^{\intercal} \mathbf{m}^{(i)}(t_{\tau,n}) \mathbf{1}^{\intercal}) \right]^2}\\
    &= \lim_{n \rightarrow \infty} \sum_{1\leq \tau\leq \ell \leq N_T} \exp(-2\lambda_1(t_{\tau,n} + t_{\ell,n}))(c^4_{\tau,\ell} + c^{2}_{\tau,\ell} + c^{0}_{\tau,\ell}),\\
\end{align*}
where 
\begin{align*}
    c^4_{\tau,\ell} &=  \expect*{\mathbf{1}\mathbf{B}^{(i)}(t_{\tau,n})^{\intercal} \mathbf{B}^{(i)}(t_{\tau,n}) \mathbf{1}^{\intercal} \mathbf{1}\mathbf{B}^{(i)}(t_{\ell,n})^{\intercal} \mathbf{B}^{(i)}(t_{\ell,n}) \mathbf{1}^{\intercal}}\\
    c^2_{\tau,\ell} &=  \expect*{\mathbf{1}\mathbf{B}^{(i)}(t_{\tau,n})^{\intercal} \mathbf{B}^{(i)}(t_{\tau,n}) \mathbf{1}^{\intercal}} \mathbf{1}\mathbf{m}^{(i)}(t_{\ell,n})^{\intercal} \mathbf{m}^{(i)}(t_{\ell,n}) \mathbf{1}^{\intercal}  \\
    c^0_{\tau,\ell} &=  \mathbf{1}\mathbf{m}^{(i)}(t_{\tau,n})^{\intercal} \mathbf{m}^{(i)}(t_{\tau,n}) \mathbf{1}^{\intercal}\mathbf{1}\mathbf{m}^{(i)}(t_{\ell,n})^{\intercal} \mathbf{m}^{(i)}(t_{\ell,n}) \mathbf{1}^{\intercal}.
\end{align*}
With Lemma \ref{lemma: Finite moment} and the fact $\lim_{t\rightarrow \infty} \exp(-2\lambda_1 t) \mathbf{1} \mathbf{m}^{(i)}(t)^{\intercal} \mathbf{m}^{(i)}(t) \mathbf{1}^{\intercal} <\infty$, we can directly see that
\begin{align*}
    \lim_{n\rightarrow \infty}& \exp(-2\lambda_1(t_{\tau,n} + t_{\ell,n})) c^0_{\tau,\ell} <\infty\\
    \lim_{n\rightarrow \infty}& \exp(-2\lambda_1(t_{\tau,n} + t_{\ell,n})) c^2_{\tau,\ell} < \infty.\\
\end{align*}
For $c^4_{\tau,\ell}$, we apply the Cauchy-Schwarz inequality to obtain
\begin{align*}
    \exp(-2\lambda_1 (t_{\tau,n} + t_{\ell,n})) c^4_{\tau,\ell} &\leq \left[\exp(-4\lambda_1 t_{\tau,n}) c^4_{\tau,n} \right]^{1/2} \left[\exp(-4\lambda_1 t_{\ell,n}) c^4_{\ell,n} \right]^{1/2},
\end{align*}
where 
\[
    c^4_{\tau,n} = \expect*{(\mathbf{1}\mathbf{B}^{(i)}(t_{\tau,n})^{\intercal}\mathbf{B}^{(i)}(t_{\tau,n})\mathbf{1}^{\intercal})^2}, \text{ for } 1\leq \tau\leq N_T.
\]
From Lemma \ref{lemma: Finite moment}, we claim that for $1\leq \tau\leq N_T$
\[\lim_{n\rightarrow \infty} \exp(-4\lambda_1 t_{\tau,n}) c^4_{\tau,n} < \infty.\]
Consequently, we have
\[\lim_{n\rightarrow \infty} \expect*{\|\mathbf{z}^{(i)}(\mathcal{T}_n)\|^4} < \infty.\]
Next, we observe that 
\[\expect*{\indi(\|\mathbf{z}^{(i)}(\mathcal{T}_n)\| \geq \epsilon \sqrt{n})} = \mathbb{P}(\|\mathbf{z}^{(i)}(\mathcal{T}_n)\|\geq \epsilon \sqrt{n}) \leq \frac{\expect*{\|\mathbf{z}^{(i)}(\mathcal{T}_n)\|}}{\epsilon \sqrt{n}}.\]
As we have already shown that $\lim_{n\rightarrow \infty}\expect*{\|\mathbf{z}^{(i)}(\mathcal{T}_n)\|^4} < \infty$, we can directly show that
\[\lim_{n \rightarrow \infty} \frac{\expect*{\|\mathbf{z}^{(i)}(\mathcal{T}_n)\|}}{\epsilon \sqrt{n}}  = 0.\]
As a result, we have equation \eqref{eq: Lindeberg Cauchy-schwarz} follows. 

We conclude this proof by applying the Lindeberg multivariate CLT:
\[\mathbf{W}_{n_i}^{(i)}(\mathcal{T}_n) = \frac{1}{\sqrt{n_i}} \sum_{j = 1}^{n_i} \mathbf{z}_{j}^{(i)}(\mathcal{T}_n) \Rightarrow N(0,\bar{\mathbf{V}}^{(i)}(\mathcal{T}_n)).\]
\end{proof}

In Proposition \ref{prop: Lindeberg CLT}, we demonstrate with an appropriate normalization factor the normalized and centered random vector $\mathbf{Y}$ is still approximately normally distributed when time depends on $n$. In practice, because we cannot employ $n\rightarrow \infty$, Proposition \ref{prop: finite CLT} is enough to obtain a good approximation of $\mathbf{Y}$.

\subsection{Details of data generation and optimization implementation in the \textit{in silico} experiments}

\label{appx: data generation and optimization implementation.}


To describe the data generation and optimization  in the \textit{in silico} experiment, we  define two distinct ranges: 
\begin{enumerate}
    \item the \textit{Generating parameters range} based on the assumption outlined in Section \ref{sec:CSC assumptions}.
    \item the \textit{Optimization feasible range} employed in all of the \textit{in silico} experiments.
\end{enumerate}
These two ranges characterize most of the differences among the \textit{in silico} experiments, while the procedure of the experiments is the same.
The detailed \textit{in silico} experiments procedure is given by the following.
We first randomly select a `true parameters set' from the \textit{Generating parameters range} to simulate the \textit{in silico} data through the Gillespie algorithm \cite{gillespie1976general}. 
Subsequently, we solve the MLE problem in equation \eqref{eq: minimize negative log likelihood} within the \textit{Optimization feasible range}. 
In particular, we employ the MATLAB Optimization Toolbox \cite{MatlabOTB} function \textit{fmincon} with sequential quadratic programming (sqp) solver. 
Given the complex likelihood equation, we solve the optimization starting from 20 initial points within the feasible region and select the estimation with the lowest negative log-likelihood. 

We denote the `true parameters sets' and the estimated parameters as $\btheta^{*}$ and $\hat{\btheta}$ respectively, while referring to the \textit{Generating parameters range} and \textit{Optimization feasible range} as $\bTheta^*$ and $\hat{\bTheta}$. 
These ranges simulate realistic scenarios and reflect our prior knowledge about nature. 
One principle in setting these two ranges is that $\bTheta^* \subseteq \hat{\bTheta}$ to simulate insufficient prior knowledge about the `true nature'.

\subsubsection{Base experiment}


In this experiment, we adhere to the assumptions outlined in Section \ref{sec:CSC assumptions}. 
In particular, the \textit{Generating parameter range} $\bTheta^*$ and the \textit{Optimization feasible range} satisfy the assumptions established in Section \ref{sec:CSC assumptions}. 
For simplicity, we implemented a special case of Assumption \ref{assump:Cell dynamic}: $\alpha_r \in (\beta_r,\beta_r + 0.1)$ and $\alpha_s = \beta_s$, (i.e. $0 = \kappa_s \leq \kappa_r \in (0,0.1)$).
Nevertheless, we also tested for the case when $\kappa_s >0$ and got a similar results.
Assumption \ref{assump:Cell plasticity} can be directly implemented in the selection of $\bTheta^*$ and $\hat{\bTheta}$, and Assumption \ref{assump:Initial structure} can be implemented in the computation of likelihood. 
In accordance with Assumption \ref{assump:Drug effect}, where $b_{r,\beta} = b_{r,\nu} = 1$, the selection of $E_{r,\beta}$ and $E_{r,\nu}$ does not impact the likelihood computation, so we set $E_{r,\beta} = E_{r,\nu} = 1$ for simplicity. 
Conversely, we carefully choose the values for $E_{s,\beta}$ and $E_{s,\nu}$ within the concentration levels implemented in the \textit{in silico} experiment to ensure the drug has an observable effect on the tumor within the given concentration levels. 


We conclude our selection of $\bTheta^*$ and $\hat{\bTheta}$ in Table \ref{tab:Illustrative range}. 
To prevent the occurrence of unreasonably large net growth rates during the optimization process, we imposed two linear constraints: $\alpha_r - \beta_r \leq 0.1$ and $\nu_{rs} + \alpha_s -\beta_s \leq 0.5$. 
In addition to the subpopulation parameters, we also selected the standard deviation parameter $c \in (0,10)$ for the observation noised added to the data.



\begin{table}[ht]
    \centering
    \begin{tabular}{|c|c|c|}
    \hline
        Parameters &  $\bTheta^*$ & $\hat{\bTheta}$ \\
    \hline
        $\alpha_r$ & $(\beta_r,\beta_r +0.1)$ & $(0,1)$\\
    \hline
        $\beta_{r}$ & $(10^{-3},0.9)$ & $(0,1)$ \\
    \hline 
        $\nu_{rs}$ & $(0,0.5)$ & $(0,1)$ \\
    \hline
        $\alpha_s$ & $\beta_s$ & $(0,1)$ \\
    \hline 
        $\beta_s$ & $(10^{-3},0.5)$ & $(0,1)$ \\
    \hline
        $b_{s,\beta}$ & $(0.8,0.9)$ & $(0.5,1)$ \\
    \hline
        $E_{s,\beta}$ & $(0.0625,2.5)$&$(0,3)$ \\
    \hline
        $b_{s,\nu}$ & $(1,1.1)$ & $(1,1.5)$ \\
    \hline
        $E_{s,\nu}$ & $(0.0625,2.5)$ & $(0,3)$ \\
    \hline
    \end{tabular}
    \caption{\textit{Generating parameter range} $\bTheta^*$ and \textit{Optimization feasible range} $\hat{\bTheta}$. According to assumptions in Section \ref{sec:CSC assumptions}, $b_{r,\beta} = E_{r,\beta} = b_{r,\nu} = E_{r,\nu} = 1$ and $\nu_{sr} = 0$ are set to both range.}
    \label{tab:Illustrative range}
\end{table}

\subsubsection{Relaxed drug effect assumption}

To model a scenario where the drug has less significant effects on the CSCs, we selected $1-b_{r,\beta} < 1 - b_{s,\beta}$ and $b_{r,\nu} - 1 < b_{s,\nu} - 1$. As in the base experiment, we present the new \textit{Generating parameters range} and \textit{Optimization feasible range} in Table \ref{tab:Relaxed drug effect range}.



\begin{table}[ht]
    \centering
    \begin{tabular}{|c|c|c|}
    \hline
        Parameters &  $\bTheta^*$ & $\hat{\bTheta}$ \\
    \hline
        $\alpha_r$ & $(\beta_r,\beta_r +0.1)$ & $(0,1)$\\
    \hline
        $\beta_{r}$ & $(10^{-3},0.9)$ & $(0,1)$ \\
    \hline 
        $\nu_{rs}$ & $(0,0.5)$ & $(0,1)$ \\
    \hline
        $b_{r,\beta}$ & $(0.97,1)$ & $(0.5,1)$ \\
    \hline
        $E_{r,\beta}$ & $(0.0625,2.5)$ & $(0,5)$\\
    \hline
        $b_{r,\nu}$ & $(1,1.03)$ & $(1,1.5)$\\
    \hline
        $E_{r,\nu}$ & $(0.0625,2.5)$ & $(0,5)$\\
    \hline
        $\alpha_s$ & $\beta_s$ & $(0,1)$ \\
    \hline 
        $\beta_s$ & $(10^{-3},0.5)$ & $(0,1)$ \\
    \hline
        $b_{s,\beta}$ & $(0.8,0.9)$ & $(0.5,1)$ \\
    \hline
        $E_{s,\beta}$ & $(0.0625,2.5)$&$(0,3)$ \\
    \hline
        $b_{s,\nu}$ & $(1,1.1)$ & $(1,1.5)$ \\
    \hline
        $E_{s,\nu}$ & $(0.0625,2.5)$ & $(0,3)$ \\
    \hline
    \end{tabular}
    \caption{\textit{Generating parameter range} $\bTheta^*$ and \textit{Optimization feasible range} $\hat{\bTheta}$ for relaxed drug effect experiments. According to assumptions in Section \ref{sec:CSC assumptions}, $\nu_{sr} = 0$ is set to both ranges.}
    \label{tab:Relaxed drug effect range}
\end{table}

\subsubsection{Relaxed initial proportion assumption}

Compared to the base experiment, we added initial proportion as variables to be estimated in these experiments.
The corresponding \textit{Generating parameters range} and \textit{Optimization feasible range} are given in Table \ref{tab:Relaxed initial proportion range}. 
It is important to note that during the optimization process, we also apply an equality constraint, $p_s = 1 - p_r$, to avoid unrealistic initial proportion. 



\begin{table}[ht]
    \centering
    \begin{tabular}{|c|c|c|}
    \hline
        Parameters &  $\bTheta^*$ & $\hat{\bTheta}$ \\
    \hline
        $p_{r}$ & $(0,1)$ & $(0,1)$\\
    \hline
        $\alpha_r$ & $\beta_r,\beta_r +0.1$ & $(0,1)$\\
    \hline
        $\beta_{r}$ & $(10^{-3},0.9)$ & $(0,1)$ \\
    \hline 
        $\nu_{rs}$ & $(0,0.5)$ & $(0,1)$ \\
    \hline
        $p_{s}$ & $1-p_r$ & $(0,1)$\\
    \hline
        $\alpha_s$ & $\beta_s$ & $(0,1)$ \\
    \hline 
        $\beta_s$ & $(10^{-3},0.5)$ & $(0,1)$ \\
    \hline
        $b_{s,\beta}$ & $(0.8,0.9)$ & $(0.5,1)$ \\
    \hline
        $E_{s,\beta}$ & $(0.0625,2.5)$&$(0,3)$ \\
    \hline
        $b_{s,\nu}$ & $(1,1.1)$ & $(1,1.5)$ \\
    \hline
        $E_{s,\nu}$ & $(0.0625,2.5)$ & $(0,3)$ \\
    \hline
    \end{tabular}
    \caption{\textit{Generating parameter range} $\bTheta^*$ and \textit{Optimization feasible range} $\hat{\bTheta}$ for relaxed initial proportion experiments. According to assumptions in Section \ref{sec:CSC assumptions}, $b_{r,\beta} = E_{r,\beta} = b_{r,\nu} = E_{r,\nu} = 1$ and $\nu_{sr} = 0$ are set to both range.}
    \label{tab:Relaxed initial proportion range}
\end{table}



\subsection{Bootstrapping technique}

\label{appx: bootstrapping}
To compute the confidence interval (CIs) of estimation, we employ the bootstrapping technique. 
This technique requires resampling the original dataset to obtain a new dataset.
In Section \ref{sec:Illustrative example}, we resampled 13 replicates from 20 generated replicates.
Then the MLE process was conducted to obtain a bootstrap estimate. 
Based on 100 bootstrap estimates, we obtained the $95\%$ CIs by taking the $2.5$ and $97.5$ percentile of those estimates.


\subsection{Parameters of interest}

\label{appx: Parameters of interest}

To evaluate the quality of the estimations, we consider the estimations of the following quantity, which can profile the phenotypical switching dynamics and the subpopulation sensitivity to the drug.

\begin{itemize}
    \item \textbf{Stable proportion}: 

    The \textit{stable proportion} of the multi-type branching process is determined by the normalized left eigenvector, $\bm{\pi}$, corresponding to the largest eigenvalue, $e^{\lambda_1 t}$, of the \textit{mean matrix}. As a result, it depends on the symmetric birth rates, death rates, and asymmetric birth rates of each subpopulation. In our drug-effect model, we assume that either the death rates or the asymmetric birth rates of each subpopulation are affected by the drug, and thus, the \textit{stable proportion} also depends on the drug concentration levels. To better understand the nature of the given tumor, we consider the \textit{stable proportion} as one important parameter of interest.

    \item $GR_{50}$:

    The $GR_{50}$, introduced in \cite{hafner2016growth}, is a newly proposed summary metric of drug-sensitivity. It is defined as the concentration at which a drug's effect on cell growth is half the observed effect. For a given maximum concentration level, $d_m$, applied during the experiment and corresponding drug effect parameters, $(b, E)$, we compute the $GR_{50}$ by
    \[GR_{50} = E\left(\frac{1-e^{r_m}}{e^{r_m} - b} \right),\]
    where $r_m = \frac{1}{2}\log\left(b + \frac{1-b}{1+\frac{d_m}{E}}\right)$ is the half-maximum effects.
    
\end{itemize}

\subsection{Detailed analysis of limited division CNSCs dynamics assumption}

\label{appx: detailed analysis of limited division CNSCs dynamic assumption}

To better understand how the our framework analyzes the data with limited CNSCs symmetric division, we monitor the relative difference (RD) between the estimation and true parameters, defined as 
\begin{equation}
    \label{eq: relative difference}
    Er(\hat{x};x^*) = \frac{\hat{x}-x^*}{x^*}.
\end{equation}
Compared to RE, RD can better depict overestimation and underestimation. 
As an analogy to the $RE = 0.2$ threshold, we consider $RD = 0.2$ and $RD = -0.2$ lines as thresholds for accurate estimation.
The RD of the estimation from 30 experiments with limited division number $G = 0,1,2,3$ are shown in Figure \ref{fig: Hierarchy G0}, Figure \ref{fig: Hierarchy G1}, Figure \ref{fig: Hierarchy G2}, and Figure \ref{fig: Hierarchy G3}, respectively.

\begin{figure}
    \centering
    \includegraphics[width = \textwidth]{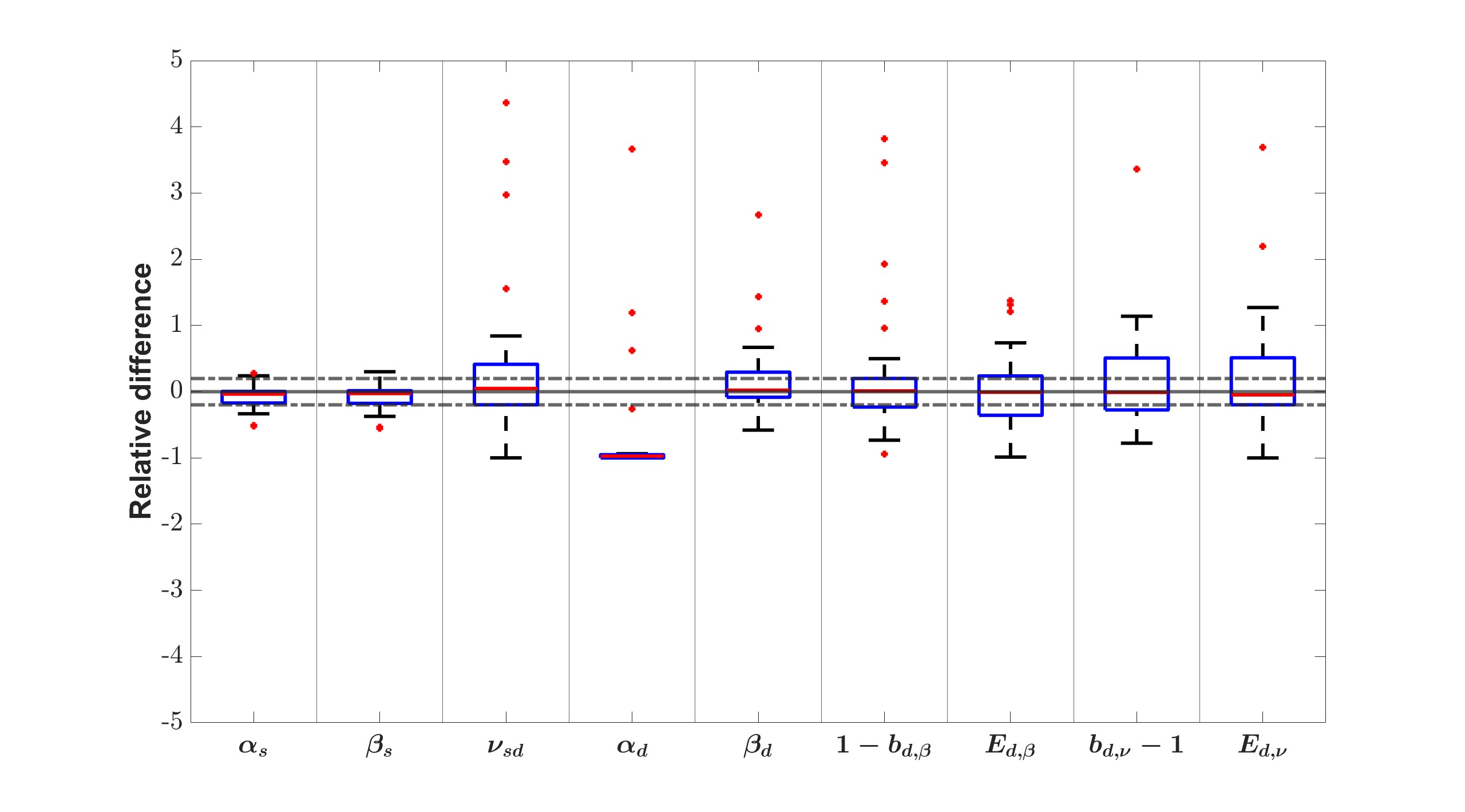}
    \caption{Relative difference of parameter estimation when CNSCs can only symmetrically divide $G = 0$ times. The solid line indicates $RD = 0$ and the dashed line indicates $RD = 0.2$ and $RD = -0.2$.}
    \label{fig: Hierarchy G0}
\end{figure}
\begin{figure}
    \centering
    \includegraphics[width = \textwidth]{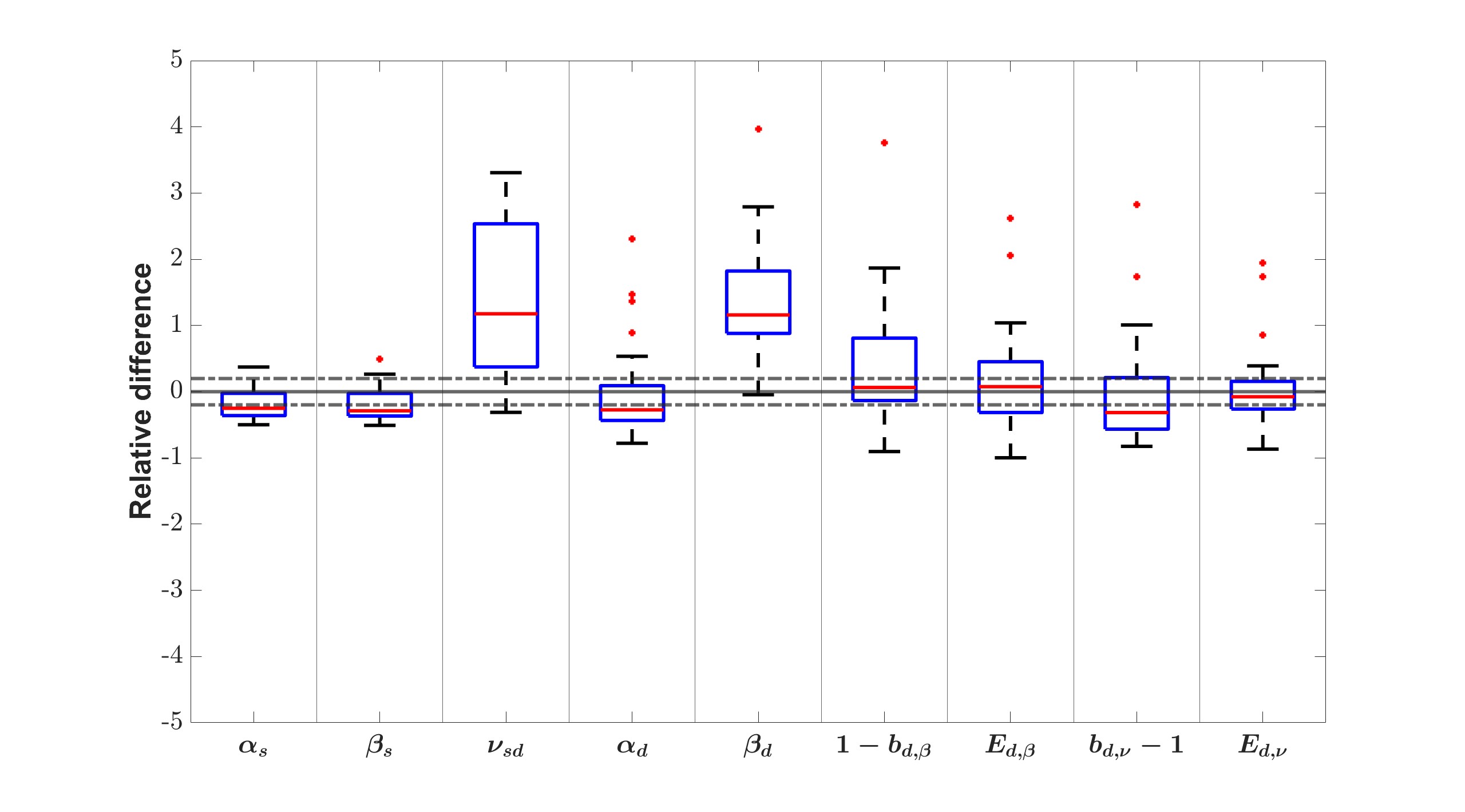}
    \caption{Relative difference of parameter estimation when CNSCs can only symmetrically divide $G = 1$ times. The solid line indicates $RD = 0$ and the dashed line indicates $RD = 0.2$ and $RD = -0.2$.}
    \label{fig: Hierarchy G1}
\end{figure}
\begin{figure}
    \centering
    \includegraphics[width = \textwidth]{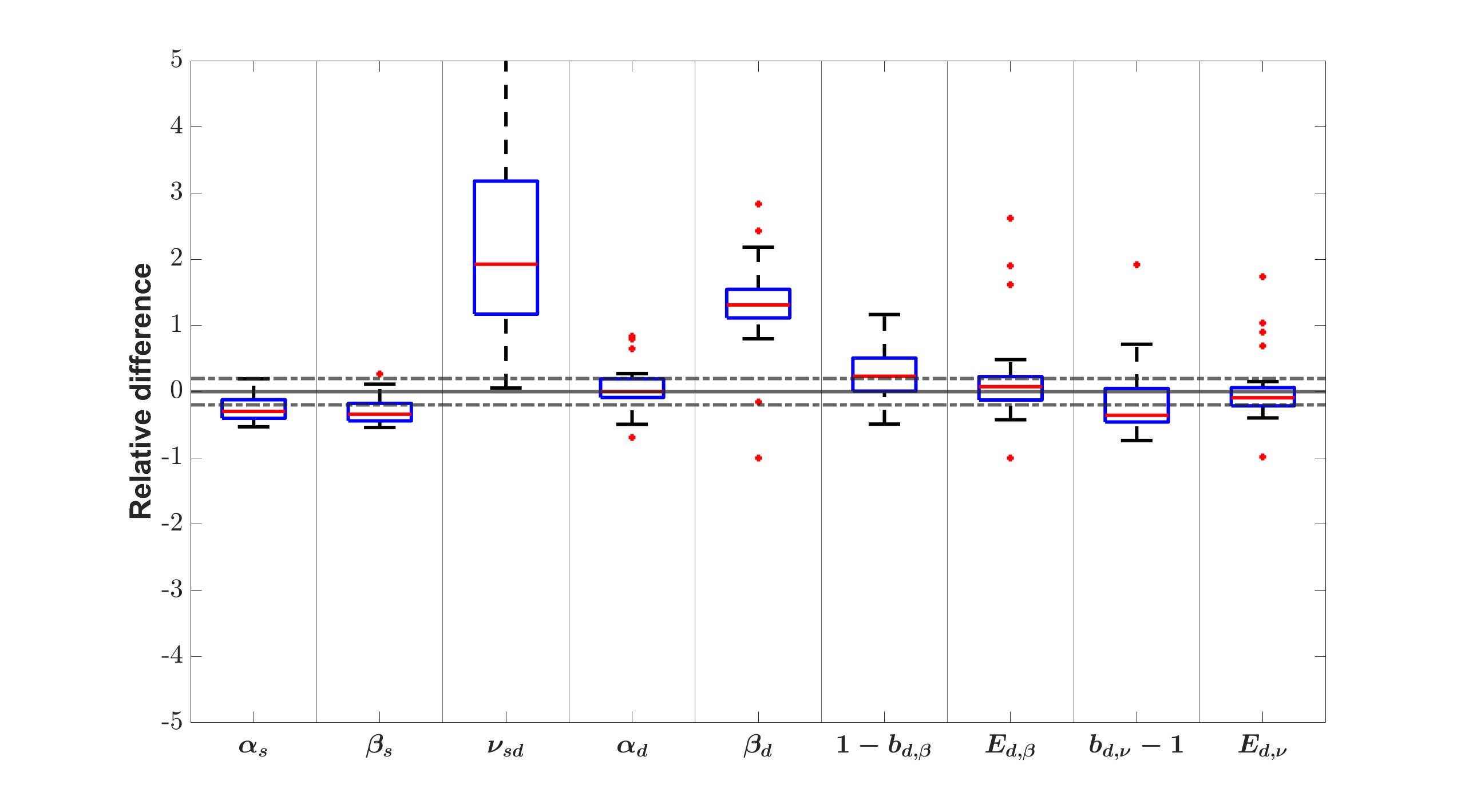}
    \caption{Relative difference of parameter estimation when CNSCs can only symmetrically divide $G = 2$ times. The solid line indicates $RD = 0$ and the dashed line indicates $RD = 0.2$ and $RD = -0.2$.}
    \label{fig: Hierarchy G2}
\end{figure}
\begin{figure}
    \centering
    \includegraphics[width = \textwidth]{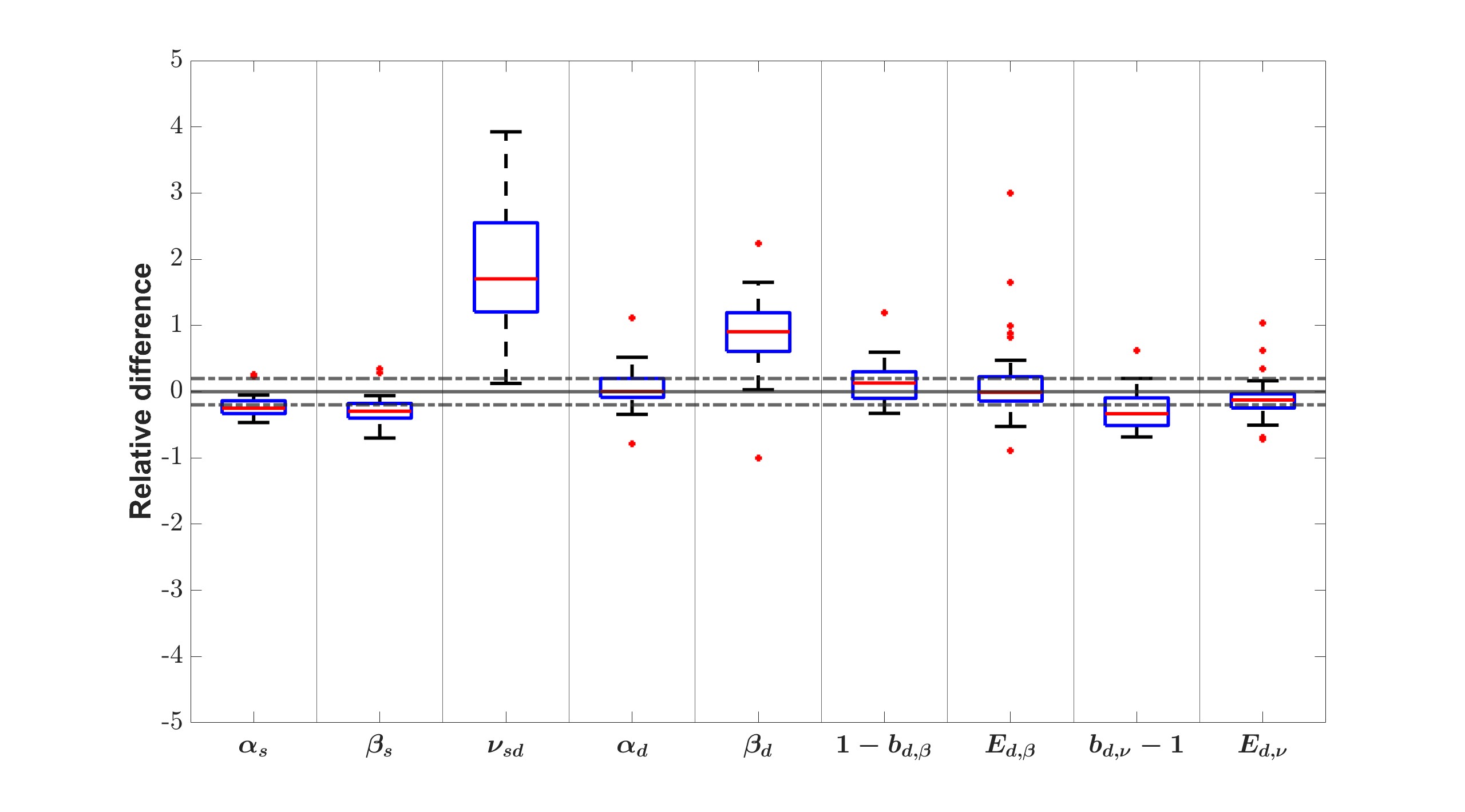}
    \caption{Relative difference of parameter estimation when CNSCs can only symmetrically divide $G = 3$ times. The solid line indicates $RD = 0$ and the dashed line indicates $RD = 0.2$ and $RD = -0.2$.}
    \label{fig: Hierarchy G3}
\end{figure}

When $G = 0$, our novel framework accurately captures the limited division cell dynamics by correctly identifying the CNSCs' symmetric division rate $\alpha_s = 0$. 
Therefore, in Figure \ref{fig: Hierarchy G0}, we observe that every parameter ehibits an $RD$ centered around $0$.  
However, there are many instances where the estimations of other parameters related to CNSCs have $RD$ value less than $-0.2$ or greater than $0.2$. 
We suggest that this could be due to the small population of CNSCs compared to CSCs when $G = 0$; hence, the estimation quality of these parameters might be influenced by the variability across distinct experiments.

In \Cref{fig: Hierarchy G1,fig: Hierarchy G2,fig: Hierarchy G3}, we observe that the estimation of $\alpha_s$ becomes more accurate with increasing $G$. 
However, the differentiation rate $\nu_{rs}$ and CNSCs' death rate $\beta_s$ are often overestimated. 
Conversely, the CSCs' symmetric division rate $\alpha_r$ and death rate $\beta_r$ tend to be slightly underestimated.  
Overall, we conclude that our framework systematically overestimates $\nu_{rs}$ and $\beta_s$, while underestimating $\alpha_r,\beta_r$ in order to better capture the $\alpha_s$.

Surprisingly, the estimation of $E_{s,\beta}$ and $E_{s,\nu}$ are not systematically overestimated or underestimated.
Additionally, these two quantities becomes more accurately estimated with increasing $G$.
Therefore, we propose two hypotheses about the estimations of $E_{s,\beta}$ and $E_{s,\nu}$:
\begin{enumerate}
    \item These estimations are not significantly affected by changes in the underlying tumor growth dynamics.
    \item The quality of these estimations might be related to the size of the targeted subpopulation, CNSCs.
\end{enumerate}
These conjectures suggest that the drug effects' inflection points $E_{s,\beta}$ and $E_{s,\nu}$ could exhibit robust estimation characteristics across different scenarios, particularly when the targeted subpopulation is abundant.



\subsection{Detailed implementation of \textit{in vitro} experiments}

In this subsection, we present a detailed implementation of the \textit{in vitro} experiment described in Section \ref{sec:In vitro}.

\subsubsection{Analysis of the \textit{AGS Time} experiment data}

\label{appx: AGS time delayed}

To understand the \textit{AGS Time} experiment data obtained from \cite{padua2023high}, we visualize the total cell count data at different time points in Figure \ref{fig:AGS Time}. It is evident that, before time point 12, the total cell count growth dynamics remains indifferent across three varying drug concentration levels. However, after 24 hours, the effect of the drug becomes significant. Therefore, we conclude that there is a time delay before CPX-O can fully activate to affect the AGS cell line. 

\begin{figure}
    \centering
    \includegraphics[width = \textwidth]{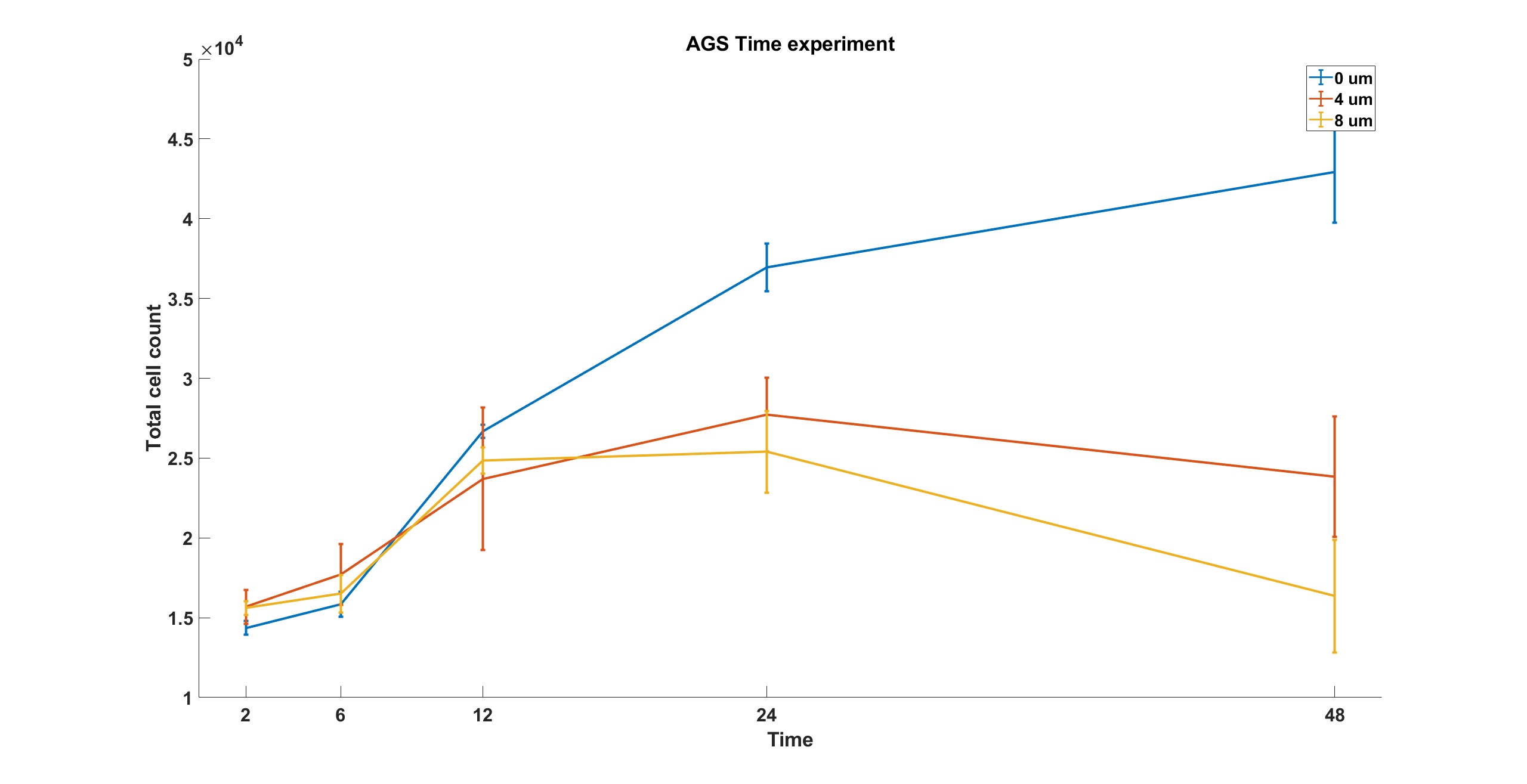}
    \caption{\textit{AGS Time} experiment data.}
    \label{fig:AGS Time}
\end{figure}

In contrast to practical observation, our statistical framework does not assume heterogeneity of the drug effect over time. As a result, our current framework struggles to adequately capture the \textit{AGS Time} experiment data without additional assumptions to incorporate the time effect on drug activation. We then consider modeling the time effect on drug activation as one of the future directions, which can further strengthen the current statistical framework. Nevertheless, we are still incorporating some aspects of the information obtained from the \textit{AGS Time} experiment.

\subsubsection{\textit{AGS Conc} TC data imputation for varying concentration levels at 12 and 24 hours.}

\label{appx: AGS Conc estimated TC}
Due to the scarcity of the \textit{AGS Time} data, we proposed to employ the simpler exponential growth model to characterize the cell growth dynamics to estimate the TC of \textit{AGS Conc} for varying concentration levels at 12 and 24 hours. In particular, this model comprises two parameters $\lambda, X_0$, representing the growth rate and initial cell count, respectively. One can mathematically formulate the cell count $X(t)$ at time $t$ by
\begin{equation}
    \label{eq: exponential growth}
    X(t,d) = X_0 \exp(\lambda(d) t),
\end{equation}
where the growth rate $\lambda$ may depend on the drug concentration level $d$.

As mentioned in the main text, the \textit{AGS Time} data comprises only 3 concentration levels, which is insufficient for estimating 4 drug-related parameters $b_{s,\beta},E_{s,\beta},b_{s,\nu},E_{s,\nu}$. Therefore, we need to utilize the \textit{AGS Conc} data, which collects data at 9 varying concentration levels. 

\textbf{Time-delayed drug effect assumption:}

In Appendix \ref{appx: AGS time delayed}, we observed that the cell growth dynamics did not change at the first 12 hours under varying doses of CPX-O. To validate our statistical framework, which does not assume a delay in the drug effect over time, we propose focusing on the data after the 12-hour checkpoint. In other words, we assume that the CPX-O does not have an effect in the first 12 hours.

\textbf{Shared underlying dynamics between \textit{AGS Time} and \textit{AGS Conc} data assumption:}

Because the \textit{AGS Time} and \textit{AGS Conc} data have different initial experimental settings, such as varying the initial total cell count and the proportion between CSCs and CNSCs, we need to assume that under no drug effect (DMSO), the underlying growth dynamics are consistent for these two experiments. Based on this assumption, we can estimate the initial TC (${\color{blue}X_{Conc}(0,0)}$) in the \textit{AGS Conc} data from TC at 48 hours ($X_{Conc}(48,0)$) based on the ratio between the initial TC (${\color{blue}X_{Time}(0,0)}$) and TC at 48 hours ($X_{Time}(48,0)$) in \textit{AGS Time} data. In other words, we assume the following equation holds
\begin{equation}
    \label{eq: Growth dynamic assumption}
    \frac{X_{Time}(48,0)}{\color{blue} X_{Time}(0,0)} = \frac{X_{Conc}(48,0)}{\color{blue} X_{Conc}(0,0)},
\end{equation}
where the colored quantity is estimated data. 

\textbf{Estimating the \textit{AGS Conc} data at 12 hour}

We estimate the ${\color{blue}X_{Time}(0,0)}$ data by fitting the exponential growth model to the DMSO \textit{AGS Time} data at $2,6,$ and $12$ time-points. After obtaining the estimation for ${\color{blue}X_{Conc}(0,0)}$, we further estimate the ${\color{blue}X_{Conc}(12,0)}$ based on the exponential growth model fitted to the DMSO \textit{AGS Time} data at $2,6,12$ time-points. Since we assume that CPX-O has no effect in the first 12 hours, we assign the same value of ${\color{blue}X_{Conc}(12,d)} = {\color{blue}X_{Conc}(12,0)} $ to all nine varying concentration levels. It is also worth mentioning that we estimate $\color{blue}X_{Conc}(12,0)$ separately for each replicate. 

\textbf{Estimating the \textit{AGS Conc} data at 24 hours}

To estimate the \textit{AGS Conc} data at $24$ hours, we first note that CPX-O affected cell growth dynamics during the time interval $(12,24$ hours). Therefore, we need to estimate the drug effect during the time interval $(12,24$ hours).
Since we only have three concentration levels, we proposed to utilize the Hill equation with two cytotoxic effects parameters $b,E$ to model the drug effect here. 
It is also worth noticing that the cytotoxic effect is the main effect that is observed. 
Based on the exponential growth model, we obtain 3 estimated growth rates $\lambda(0),\lambda(4),\lambda(8)$ for 3 concentration levels of CPX-O.
We then fit these growth rates through the 2 parameters Hill equations
\begin{equation*}
    \label{eq: exponential Hill}
    \lambda(d) = \lambda(0) + \log\left(b + \frac{1-b}{1+d/E}\right).
\end{equation*}
With the estimated parameters $b,E$, we estimate the growth rate $\lambda(d), d \in \{0,0.125,0.25,0.5,1,2,4,8,16\}.$
We then plug in these estimated growth rates to estimate the TC at 24 hours, i.e. 
$$
X(24,d) = X(12,0)\exp(12 \lambda(d)), d\in \{0,0.125,0.25,0.5,1,2,4,8,16\}.
$$

\subsubsection{Details of maximum likelihood estimation when AIC computation}

\label{appx:In vitro details}

\textbf{AGS--CPX-O AIC computation:}

To calculate the AIC value, we first need to obtain the maximum likelihood estimate for each model.
Following the experimental setup described in Appendix \ref{appx: data generation and optimization implementation.}, Table \ref{tab:AGS SORE6 range} provides an \textit{Optimization feasible range} used in the experiment. 
The same optimization range is applied to the free parameters across all four model assumptions.

\begin{table}[ht]
    \centering
    \begin{tabular}{|c|c|}
    \hline
        Parameters  & $\hat{\bTheta}$ \\
    \hline
        $\alpha_r$  & $(\beta_r,\beta_r+0.1)$\\
    \hline
        $\beta_{r}$  & $(0,1)$ \\
    \hline 
        $\nu_{rs}$  & $(0,1)$ \\
    \hline
        $\alpha_s$  & $(\beta_s,\beta_s+0.1)$ \\
    \hline 
        $\beta_s$  & $(0,1)$ \\
    \hline
        $b_{s,\beta}$  & $(0,1)$ \\
    \hline
        $E_{s,\beta}$ &$(0,16)$ \\
    \hline
        $m_{s,\beta}$ & $(0.1,5)$\\
    \hline
        $b_{s,\nu}$  & $(1,2)$ \\
    \hline
        $E_{s,\nu}$  & $(0,16)$ \\
    \hline
        $m_{s,\nu}$  & $(0.1,5)$\\
    \hline
        $k$ & $(0,1)$\\
    \hline
        $t_0$ & $(0,36)$\\
    \hline
        $c$ & $(0,300)$\\
    \hline
    \end{tabular}
    \caption{\textit{Optimization feasible range} $\hat{\bTheta}$ for \textit{in vitro} drug-induced plasticity experiment. }
    \label{tab:AGS SORE6 range}
\end{table}

\textbf{COLO858--Vemurafenib AIC computation:}

To computing the AIC value, we calculated the maximum likelihood estimate for each model based on the \textit{Optimization feasible range} described in Table {\ref{tab:COLO858--Vemurafenib range}}.

\begin{table}[ht]
    \centering
    \begin{tabular}{|c|c|}
    \hline
        Parameters  & $\hat{\bTheta}$ \\
    \hline
        $\alpha_s$  & $(0,0.1)$\\
    \hline
        $\beta_{s}$  & $(0,0.1)$ \\
    \hline 
        $\nu_{sr}$  & $(0,0.5)$ \\
    \hline
        $\alpha_r$  & $(0,0.1)$ \\
    \hline 
        $\beta_r$  & $(0,0.1)$ \\
    \hline
        $\nu_{rs}$ & $(0,0.5)$\\
    \hline
        $b_{r/s,\beta}$  & $(0,1)$ \\
    \hline
        $E_{r/s,\beta}$ &$(0,50)$ \\
    \hline
        $b_{s,\nu}$  & $(1,2)$ \\
    \hline
        $E_{s,\nu}$  & $(0,50)$ \\
    \hline
        $k$ & $(0,5)$\\
    \hline
        $t_0$ & $(0,36)$\\
    \hline
        $c$ & $(0,50)$\\
    \hline
    \end{tabular}
    \caption{\textit{Optimization feasible range} $\hat{\bTheta}$ for COLO858--Vemurafenib validation experiment. The subscripts $s$ and $r$ represent the sensitive and resistant subpopulations, respectively.}
    \label{tab:COLO858--Vemurafenib range}
\end{table}

    \bibliographystyle{plain}
    \bibliography{main.bib}

\end{document}